\documentclass[journal,10pt]{IEEEtran}

\usepackage[nolist]{acronym}
\begin{acronym}

 \acro{RIS}{reconfigurable intelligent surface}
 \acro{FRIS}{Fluid Reconfigurable Intelligent Surface}
 \acro{FAS}{Fluid Antenna System}
 \acro{SNR}{signal-to-noise ratio}
 \acro {BS}{base station}
 \acro{UE}{user equipment}
 \acro{UPA}{uniform planar array}
 \acro{PDF}{probability density function}
\acro{FRISs}{fluid reconfigurable intelligent surfaces}
\acro{CDF}{cumulative distribution function}

\end{acronym}

\usepackage{amssymb,amsmath,amsfonts,amsthm,bm,bbm}
\usepackage{empheq}
\usepackage{algorithm}
\usepackage[noend]{algpseudocode} % provides \State, \For, \If, ...

\usepackage[caption=false,font=normalsize,labelfont=sf,textfont=sf]{subfig}
\usepackage{graphicx}
\usepackage{mathrsfs}

\usepackage{array}
\usepackage{stfloats}
\usepackage{url}
\usepackage{verbatim}

\usepackage{cite}

\usepackage{float}

\usepackage[colorlinks=true, linkcolor=blue, citecolor=blue, urlcolor=blue, hypertexnames=false, breaklinks=true, linktocpage]{hyperref}

\newtheorem*{theorem}{Theorem}

\newtheorem{corollary}{Corollary}

\DeclareMathOperator{\ii}{\textnormal{{i}}} 
\usepackage{soul}

\begin{document}

\title{Exact Statistical Characterization and Performance Analysis \\ of Fluid Reconfigurable Intelligent Surfaces}

% \title{Statistical Characterization and Performance Analysis of Fluid Reconfigurable Intelligent Surfaces: An Exact Framework}

% \title{On the Statistics of and Performance Evaluation of \\ Fluid Reconfigurable Intelligent Surfaces}

% \title{On the Operating Regime Statistics of Fluid \\ Reconfigurable Intelligent Surfaces}

% \title{On the Statistics and Operating Regimes of  Fluid \\ Reconfigurable Intelligent Surfaces}

% \title{Exact Statistical Characterization and Performance Analysis of \\ Fluid Reconfigurable Intelligent Surfaces}

% \title{Exact  Outage Analysis and Statistical Modeling of Fluid Reconfigurable Intelligent Surfaces}

% \title{Fluid Reconfigurable Intelligent Surfaces: \\ An Exact Characterization and Performance Analysis}

\author{Masoud Khazaee, Felipe A. P. de Figueiredo, Rausley A. A. de Souza, \IEEEmembership{Senior Member,~IEEE}, \\ Farshad Rostami Ghadi, \IEEEmembership{Member,~IEEE}, Kai-Kit Wong, \IEEEmembership{Fellow,~IEEE}, Luciano L. Mendes,~%
\IEEEmembership{Member,~IEEE}, \\ and Fernando D. Almeida García,~\IEEEmembership{Senior Member,~IEEE}
\thanks{This work was funded by the Brasil 6G Project with support from RNP/MCTI (Grant 01245.010604/2020-14), and by the xGMobile Project (XGM-AFCCT-2025-8-1-1 and XGM-AFCCT-2024-9-1-1) with resources from EMBRAPII/MCTI (Grant 052/2023 PPI IoT/Manufatura 4.0), by CNPq (302085/2025-4, 306199/2025-4, 141545/2024-0), by FAPEMIG (APQ-03162-24, PPE-00124-23, RED-00194-23), and by FINEP (nº 1060/2 contract 01.25.0883.00).
}%
\thanks{M. Khazae and L. L. Mendes  are with the Radiocommunications
Reference Center, National Institute of Telecommunications (INATEL), Santa Rita do Sapucaí, MG,  37536-001, Brazil (e-mail: \mbox{masoud.khazaee@dtel.inatel.br}; \mbox{lucianol@inatel.br}).
Felipe A. P. de Figueiredo, Rausley A. A. de Souza, and F. D. A. García are with the Wireless and Artificial Intelligence Laboratory (WAI Lab), National Institute of Telecommunications (INATEL), Santa Rita do Sapucaí, MG,  37536-001, Brazil (e-mail: \mbox{felipe.figueiredo@inatel.br}; \mbox{rausley@inatel.br}; \mbox{fernando.almeida@inatel.br}).

F. Rostami Ghadi and K.-K. Wong are with the Department of Electronic and Electrical
Engineering, University College London, Torrington Place, WC1E 7JE, U.K. (e-mail: \mbox{f.rostamighadi@ucl.ac.uk}; \mbox{kai-kit.wong@ucl.ac.uk}).
K.-K. Wong is also with the Department of Electronic Engineering, Kyung Hee University, Yongin-si, Gyeonggi-do 17104, Republic of Korea.
}
}

% \markboth{IEEE Transactions on Wireless Communications}%
% {Shell \MakeLowercase{\textit{et al.}}: }

%\IEEEpubid{~\copyright}

\maketitle

\begin{abstract}
Fluid reconfigurable intelligent surfaces (FRIS) extend conventional RIS architectures by enabling physical reconfiguration of element positions, thereby introducing a fundamentally new degree of freedom for controlling spatial correlation and improving link reliability. 
Despite this promise, rigorous performance analysis of FRIS-assisted wireless systems has remained challenging, as exact statistical analyses of the end-to-end cascaded channels have been unavailable.
This paper addresses this gap by providing the first \emph{exact closed-form} characterization of the end-to-end cascaded channel gain in FRIS-aided systems under general spatial correlation. By exploiting the spectral structure of the FRIS-induced correlation matrix, we show that the channel gain statistics can be represented as a finite linear combination of $\boldsymbol{K}$-distributions. This unified formulation naturally captures fully correlated, effectively decorrelated, and intrinsically uncorrelated operating regimes as special cases.
Building on the derived channel statistics, we further obtain exact closed-form expressions for the outage probability and ergodic capacity. We also conduct an outage-based asymptotic analysis, which reveals the \emph{true} diversity order of the system. Numerical results corroborate the proposed analytical framework via Monte Carlo simulations, benchmark its accuracy against state-of-the-art approximation-based approaches, and demonstrate that fluidic reconfiguration can yield tangible reliability gains by reshaping the spatial correlation structure.
\end{abstract}

\begin{IEEEkeywords}
Fluid reconfigurable intelligent surface (FRIS), spatial correlation, probability density function, cumulative distribution function, ergodic capacity, outage probability.
\end{IEEEkeywords}

\section{Introduction}

\IEEEPARstart{R}{econfigurable} intelligent surfaces (RIS) have attracted considerable attention as a cost-effective means to enhance wireless communication systems by intelligently shaping the propagation environment through a large number of nearly passive reflecting elements \cite{10596064}. By appropriately configuring these elements, RIS can improve wireless coverage, enhance energy efficiency, and increase spectrum utilization~\cite{9475155}. 
% In particular, as the number of reflecting elements grows, channel estimation and control overhead become increasingly complex. Moreover, the passive nature of RISs leads to a multiplicative path loss across the cascaded links, which may result in severe signal attenuation. 
Despite this promise, several fundamental challenges hinder the practical deployment of RIS \cite{9765815}. In particular, the close physical spacing of RIS elements often induces strong spatial correlation, which degrades the achievable system performance. Moreover, most existing RIS implementations rely on physically rigid architectures, inherently limiting their ability to exploit spatial diversity---even in scenarios where sufficient physical space is available for reconfiguration.

To address the above RIS limitations, greater spatial adaptability can be introduced by incorporating the fluid antenna system (FAS) paradigm \cite{9264694, 10753482,10909643,11247926,11302793,11175437}. FAS represents a new class of reconfigurable antennas that enable dynamic adjustment of both antenna shape and position, thereby unlocking additional spatial degrees of freedom at the physical layer \cite{yang2025bridgingtheorypracticereconfigurable}. By leveraging controlled physical movement within a confined aperture, FAS-based architectures can effectively mitigate spatial correlation without requiring complex signal processing.
Building upon this concept, fluid reconfigurable intelligent
surfaces (FRIS) \cite{salem2025lookperformanceenhancementpotential} and the more general framework of fluid integrated reflecting and emitting surfaces (FIRES) \cite{11137368} have recently been proposed to address the rigidity of conventional RIS structures by enabling physical reconfiguration of the element positions themselves.
 This fluidic capability allows the system to dynamically select subsets of elements exhibiting reduced mutual correlation, without relying on sophisticated phase-shift optimization. As a result, FRIS offer a fundamentally different mechanism for improving link reliability---one that exploits spatial decorrelation through physical reconfiguration rather than purely electromagnetic tuning---thereby providing a scalable and hardware-efficient alternative for next-generation intelligent surfaces.

\subsection{Related Works}

Although a few preliminary studies have recently appeared, the analytical characterization and performance optimization of RIS-aided fluid antenna systems and FRIS-assisted systems remain at a nascent stage
\cite{11075830 ,ghadi2025coverageanalysisoptimizationfiresassisted,11368654,10978677,10539238,11154019Ghadi,Ghadi2025FRISCovert}. The primary difficulty stems from
deriving the fundamental statistical descriptors---namely, the \ac{PDF} and \ac{CDF}---of the
end-to-end cascaded channel gain induced by RIS and FRIS architectures.
Due to the inherent product and sum–product structures of the cascaded links, obtaining exact closed-form statistics is mathematically challenging.
Consequently, most existing works have relied on approximation-based
approaches. For example, \cite{10978677} approximated the cascaded channel statistics of a RIS-aided fluid antenna system using the central limit theorem (CLT) in conjunction with copula theory, and employed the resulting model to evaluate the secrecy outage probability (SOP). Similarly, \cite{10539238} adopted a CLT- and copula-based approximation to characterize the cascaded channel of a RIS-aided fluid antenna system, which was subsequently used to analyze system performance in terms of
outage probability (OP) and delay outage rate (DOR).
More recently, \cite{11154019Ghadi} approximated the PDF and CDF of the FRIS
cascaded channel by those of a Gamma distribution via a moment-matching
approach, enabling the evaluation of OP and ergodic capacity (EC). Building
upon this approximation, \cite{Ghadi2025FRISCovert} investigated FRIS-assisted
covert communications, analyzing performance metrics such as covertness outage
probability (COP), OP, and success probability (SP).

Despite their valuable insights, the aforementioned works predominantly rely on
asymptotic arguments or distributional approximations. Such approaches may
mask the precise effects of spatial correlation and fluidic reconfiguration on system performance. This observation highlights
the need for exact statistical characterizations of FRIS-assisted cascaded
channels---an open problem in the literature that directly motivates the
analytical framework developed in this work.

\subsection{Contributions}

Motivated by the above discussion, the main contributions of this work are
summarized as follows:
\begin{enumerate}
    \item By means of a rigorous spectral decomposition analysis, we derive
    \emph{exact closed-form} expressions for the PDF and CDF of the end-to-end cascaded
    channel gain of a FRIS-assisted wireless system operating under a general spatial
    correlation regime. The resulting expressions are given as finite linear
    combinations of $K$-distributions. 
    % Notably, the proposed framework applies to both fixed (deterministic) FRIS phase configurations and uniformly random FRIS phase configurations. 
    To the best of the authors’ knowledge, these results have not been previously reported in the literature, including for conventional correlated RIS-assisted systems.

    \item As special cases of the general framework, we further derive exact closed-form expressions for the PDF and CDF under several practically relevant spatial correlation regimes, including effectively decorrelated, fully correlated, and intrinsically uncorrelated scenarios. Each regime is accompanied by a clear physical interpretation that highlights the impact of spatial correlation and fluidic reconfiguration on the resulting channel statistics.

    \item Leveraging the derived channel statistics, we conduct a comprehensive performance analysis. Specifically, we derive novel exact and asymptotic closed-form expressions for the OP, which rigorously establish that the \emph{true} diversity order of the system is equal to one. This result refines and clarifies the diversity behavior reported in \cite[Corollary~1]{11154019Ghadi}.
    % where the diversity order was inferred from an approximation-based Gamma model and therefore tied to the shape parameter of the approximating distribution.
    In addition, we derive new exact closed-form expressions for the EC, thereby enabling a complete and analytically tractable performance evaluation of FRIS-assisted systems.

    \item Extensive numerical simulations are conducted to validate the
    analytical results and to illustrate the impact of fluidic
    reconfiguration and spatial correlation on system performance. The accuracy
    of the proposed analytical framework is further corroborated through
    comparisons with the state-of-the-art approximation in
    \cite{11154019Ghadi}.
\end{enumerate}

% This work derives an \textbf{exact closed-form expression} for the probability density function (PDF) of the equivalent channel gain 
% $G = |H_{\text{eq}}|^2$, expressed as a finite linear combination of $K$-distributions. 
% This explicit formulation, obtained from the quadratic form of correlated complex Gaussian variables, provides a complete statistical characterization of the FRIS-assisted channel, which was not presented in prior studies \cite{11154019Ghadi} that relied solely on Gamma-based approximations.

\subsection{Structure and Notation}

The remainder of this paper is organized as follows. Section~\ref{sec: System Model} describes the
considered FRIS-assisted system model and introduces the underlying channel
assumptions. Section~\ref{sec: Channel Gain Statistics} presents the exact statistical characterization of the
end-to-end cascaded channel gain under general spatial correlation, including
the main theorem and its corollaries. In Section~\ref{sec: Performance Analysis}, the derived statistics are
used to analyze key performance metrics, including the outage probability and
ergodic capacity. Numerical results and discussions are provided in
Section~\ref{sec: Numerical Results} to validate the analytical findings and illustrate the impact of
fluidic reconfiguration and spatial correlation. Finally, Section~\ref{sec: Conclusions} concludes
the paper.

% Throughout the paper, the following notation is used. 
\textit{Notation}: In the sequel, $\Pr(\cdot)$ denotes
probability; $\mathbb{E}[\cdot]$, expectation; $J_0(\cdot)$,
the zeroth-order Bessel function of the first kind; $K_{\nu}(\cdot)$,
the modified Bessel function of the second kind of order $\nu$;
$(\cdot)^\text{T}$ and $(\cdot)^\text{H}$ denote the transpose and Hermitian transpose,
respectively; $|\cdot|$, the absolute value; $\|\cdot\|_{2}$,
the Euclidean norm; $\Gamma(\cdot)$, the Gamma function; and the symbol
$\sim$ means ``distributed as.'' The operator $\mathrm{diag}(\cdot)$ denotes a
diagonal matrix; $\mathbb{N}$, the set of positive integers (excluding
zero); and $\mathbb{C}$, the set of complex numbers.
The notation $\mathcal{G}(a,b)$ represents a Gamma distribution with shape
parameter $a$ and scale parameter $b$; $\mathcal{CN}(a,b)$, a circularly
symmetric complex Gaussian distribution with mean $a$ and variance $b$; and
$\mathrm{Exp}(a)$, an exponential distribution with rate parameter $a$; $\mathcal{U}[a,b)$, a uniform distribution over the
interval $[a,b)$; $\mathbf{I}_{r}$, the $r\times r$ identity matrix; and
$\ii=\sqrt{-1}$, the imaginary unit. Finally, the symbol $\simeq$ denotes
asymptotic equivalence near zero, i.e.,
$h(x)\simeq g(x)
\iff
\lim_{x\to 0}\frac{h(x)}{g(x)}=1$.

\section{System Model}
\label{sec: System Model}

We consider a downlink wireless communication scenario consisting of a
\ac{BS}, a FRIS, and a 
\ac{UE}. The direct \ac{BS}–\ac{UE} path is assumed to be completely 
blocked by surrounding obstacles, so all signal propagation occurs 
exclusively through the surface, as depicted in Fig.~\ref{fig:model}.

\subsection{Array Geometry}

The FRIS is modeled as a \ac{UPA} with
$M = M_x M_z$ passive elements arranged along the horizontal ($x$) and
vertical ($z$) axes.
The inter-element spacing is given by $d = d_w \lambda_c$, where $\lambda_c$ denotes the carrier wavelength and $d_w$ is the normalized (dimensionless) element spacing.
\begin{figure}[t!]
\centering
\includegraphics[width=0.6\linewidth]{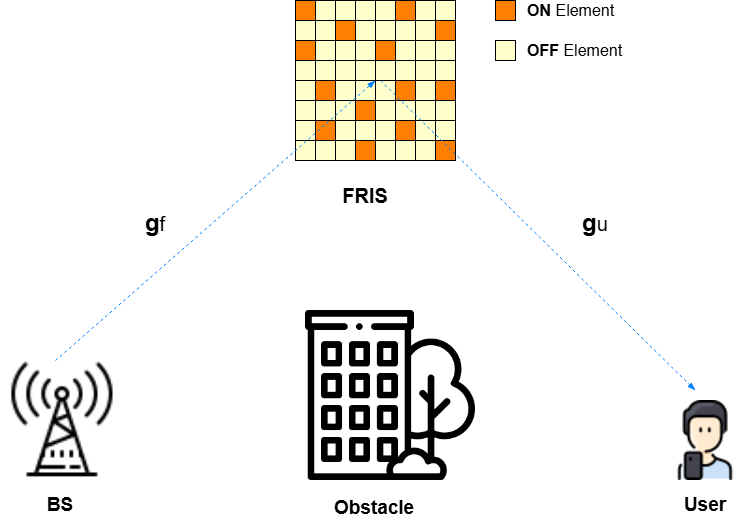}
\caption{System model of the FRIS-aided communication system.}
\label{fig:model}
\end{figure}

The physical location of the $(i,j)$-th reflecting element is
\begin{equation}
\mathbf{p}_{i,j}
=
\big[ i\, d_w \lambda_c,\; j\, d_w \lambda_c \big]^\text{T},
\end{equation}
for $i = 0,\ldots,M_x - 1$ and $j = 0,\ldots,M_z - 1$.

% \[
% \mathbf{p}_{i,j}
% =
% \big[i\, d_\lambda \lambda,\; j\, d_\lambda \lambda \big]^{\top},
% \qquad
% i = 0,\ldots,M_x - 1,\;
% j = 0,\ldots,M_z - 1 .
% \]
For bookkeeping purposes (e.g., active-set selection), it is convenient to 
associate each pair $(i,j)$ with a global linear index
\begin{equation}
m = i + j M_x + 1,
\qquad m = 1,\ldots,M ,
\label{eq:linear_index}
\end{equation}
which enumerates the \ac{UPA} in row-major order.  
This index is used only to label elements; the physical coordinates remain 
represented by $\mathbf{p}_{i,j}$.
Although $M_z$ does not appear explicitly in \eqref{eq:linear_index},
its effect is fully included through the range $j=0,\ldots,M_z-1$, which
determines how many rows the indexing spans.

\subsection{Spatial Correlation and Active-Set Selection}
\label{activation}

In the FRIS architecture, the close physical spacing of adjacent elements 
induces non-negligible spatial correlation, which must be accounted for in the 
system modeling. Throughout this work, spatial correlation is characterized using 
the classical Clarke--Jakes model~\cite{Stuber1996MobileComm}:
\begin{equation}
\rho\!\left(r_{(i,j),(i',j')}\right)
=
J_0\!\left(\frac{2\pi}{\lambda_c}\, r_{(i,j),(i',j')}\right),
\label{eq:jakes_corr}
\end{equation}
where  $r_{(i,j),(i',j')}
\triangleq
\big\|\mathbf{p}_{i,j}-\mathbf{p}_{i',j'}\big\|_2$ is the Euclidean distance between two \ac{UPA} elements located at coordinates 
$(i,j)$ and $(i',j')$.
Since the \ac{UPA} forms a regular grid with inter-element spacing $d$, the distance 
between two elements separated by $p$ horizontal and $q$ vertical indices is
\begin{align}
\label{eq: distance d}
\big\|\mathbf{p}_{i,j}-\mathbf{p}_{i+p,j+q}\big\|_2
= d\sqrt{p^{2}+q^{2}}.
\end{align}

Let $M_{\mathrm{on}}$ denote the number of active elements, and let 
$\mathcal{S}\subseteq\{1,\ldots,M\}$ denote their set of linear indices.
The active-set selection criterion is based on the \emph{worst-case local 
neighbor correlation}~\cite{Sanayei2004AntennaSelection}.\footnote{As 
discussed in~\cite{11075830}, FRIS element selection may also be carried out 
via a binary activation vector jointly optimized with discrete phase shifts 
using cross-entropy–based sampling methods.}  
Specifically, using~\eqref{eq:jakes_corr}, the design constraint is

\begin{equation}
\max_{(p,q)\in\mathcal N}
\left| 
J_0\!\left(
\tfrac{2\pi s d_w}{\lambda_c}\sqrt{p^2+q^2}
\right)
\right|
\;\le\; \tau ,
\label{eq:cap}
\end{equation}
where the stencil
\[
\mathcal N
=
\{(1,0),\, (1,1),\, (2,0),\, (2,1),\, (2,2),\, (3,0)\}
\]
includes the dominant axial and diagonal offsets corresponding to nearest and next-nearest neighbors.
In addition, $\tau\in(0,1)$ is a target correlation threshold, and $s\in\mathbb{N}$ 
denotes the \emph{stride}---a design parameter that specifies how far apart the active FRIS elements are placed within the \ac{UPA} grid. As $s$ grows, the grid (and thus physical) separation between active elements grows, leading to lower spatial correlation.

For each selected index $m\in\mathcal{S}$, let $(i,j)$ denote the associated 
\ac{UPA} coordinates obtained via~\eqref{eq:linear_index}. The spatial correlation matrix restricted to the active set is then
\begin{equation}
\mathbf{R}_{\mathcal{S}}
=
\Big[
J_0\!\left(
\frac{2\pi}{\lambda_c},
\big\lVert \mathbf{p}_{i,j}-\mathbf{p}_{i',j'} \big\rVert_2
\right)
\Big]_{m,m'\in\mathcal{S}}
\in \mathbb{C}^{M_{\mathrm{on}}\times M_{\mathrm{on}}},
\label{eq:RS-def}
\end{equation}
whose entries represent pairwise Clarke--Jakes spatial correlations between the 
active FRIS elements.  
Owing to the properties of the isotropic scattering model (see Section~\ref{sec: Received Signal Model}) and the even symmetry 
of the Bessel function $J_0(\cdot)$, the matrix $\mathbf{R}_{\mathcal{S}}$ is 
Hermitian and positive semidefinite~\cite[Sec.~2.4]{Stuber1996MobileComm}.

\subsection{Received Signal Model}
\label{sec: Received Signal Model}
% \subsection{Received Signal and FRIS Setup}
% \subsection{Received Signal Model and FRIS Configuration}
% \subsection{Received Signal and Small-Scale Fading}

Under the considered wireless setup, the received baseband signal at the UE can be expressed as\footnote{It is worth mentioning that although the activation procedure for the FRIS elements differs from that in \cite[eq.~(3)]{11154019Ghadi}, the resulting received baseband signals are equivalent.}
\begin{equation}
y = \sqrt{P L_\text{f} L_\text{u}}\;
\mathbf{g}_\text{u}^\text{H}\,
\mathbf{R}_{\mathcal{S}}^{1/2}\,
\boldsymbol{\Phi}\,
\mathbf{R}_{\mathcal{S}}^{1/2}\,
\mathbf{g}_\text{f}\, x + z,
\label{eq:rx-final}
\end{equation}
where $P$ denotes the transmit power, $x$ is the unit-energy transmitted symbol, and  
$z \sim \mathcal{CN}(0,N_0)$ represents additive white Gaussian noise (AWGN) with noise power spectral density $N_0$.  
The large-scale attenuation factors associated with the \ac{BS}--FRIS and FRIS--\ac{UE} links are given by
$L_\text{f} = \rho_\text{f} d_\text{f}^{-\alpha_\text{f}}$ and $L_\text{u} = \rho_\text{u} d_\text{u}^{-\alpha_\text{u}}$, respectively, where $\alpha_\text{f}$ and $\alpha_\text{u}$ are the pathloss exponents and $d_\text{f}$ and $d_\text{u}$ denote the corresponding link distances.
The vectors $\mathbf{g}_\text{f}, \mathbf{g}_\text{u} \sim \mathcal{CN}(\mathbf{0},\mathbf{I}_{M_{\mathrm{on}}})$ model independent small-scale fading on the BS--FRIS and FRIS--UE links, respectively, and $\mathbf{R}_{\mathcal{S}}$ is the spatial correlation matrix defined in~\eqref{eq:RS-def}.  
Moreover, $\mathbf{g}_\text{f}$ and $\mathbf{g}_\text{u}$ are statistically independent, and $\boldsymbol{\Phi} = \mathrm{diag}\!\left(e^{\ii \theta_1},\ldots,e^{\ii \theta_{M_{\mathrm{on}}}}\right)$ denotes the FRIS reflection matrix, whose diagonal entries impose the element-wise phase shifts. The vector of reflection phases is 
$\boldsymbol{\theta}\triangleq[\theta_1,\ldots,\theta_{M_{\mathrm{on}}}]^\text{T}$, with each $\theta_\ell \in [0,2\pi)$ representing the $\ell$-th phase shift for $\ell \in \{1,\ldots,M_{\mathrm{on}}\}$.

% The diagonal matrix $ \boldsymbol{\Phi} = \mathrm{diag}\left(e^{j\theta_1},\ldots,e^{j\theta_{M_{\mathrm{on}}}}\right)$
% contains the FRIS-induced phase shifts, where $\theta_\ell \in [0,2\pi)$ denotes the $\ell$-th reflection phase.

Observe from~\eqref{eq:rx-final} that the equivalent correlated channels are given by
$\tilde{\mathbf{g}}_\text{f} \triangleq \mathbf{R}_{\mathcal{S}}^{1/2}\mathbf{g}_\text{f}$ and
$\tilde{\mathbf{g}}_\text{u} \triangleq \mathbf{R}_{\mathcal{S}}^{1/2}\mathbf{g}_\text{u}$ for the BS--FRIS and FRIS--UE links, respectively.

For compactness, define the composite coupling matrix and the end-to-end channel gain as, respectively, 
\begin{align}
    \mathbf{A}
    &\triangleq 
    \mathbf{R}_{\mathcal{S}}^{1/2}\,
    \boldsymbol{\Phi}\,
    \mathbf{R}_{\mathcal{S}}^{1/2}, 
    \label{eq:A-theta}
    \\
    G_0 
    &\triangleq 
    \big|\mathbf{g}_\text{u}^\text{H}\mathbf{A}\mathbf{g}_\text{f}\big|^2.
    \label{eq:G0-def}
\end{align}

% --- Add after Eq. (8) / after defining A and C ---

In \eqref{eq:A-theta}, the phase-shift matrix $\boldsymbol{\Phi}$ is assumed to be arbitrary but deterministic.
Therefore, all derived distributions and performance metrics are understood as being conditioned on 
$\boldsymbol{\Phi}$. 
% The only randomness is due to the small-scale fading vectors $\mathbf{g}_{\mathrm{f}}$ and $\mathbf{g}_{\mathrm{u}}$.

% \subsection{Instantaneous SNR, Outage Probability, and Ergodic Capacity}

\subsection{Definitions and Performance Metrics}

\subsubsection{Instantaneous SNR}
From \eqref{eq:rx-final}, the instantaneous signal-to-noise ratio (SNR) at the UE is given by
\begin{equation}
    \gamma = \bar{\gamma}\, L_\text{f} L_\text{u}\, G_0,
    \label{eq:SNR-def}
\end{equation}
where $\bar{\gamma} \triangleq P/N_0$ denotes the normalized average received SNR.

\subsubsection{Instantaneous Channel Capacity}
The instantaneous channel capacity, measured in bits/s/Hz, is defined as
\begin{equation}
    C \triangleq \log_2\!\left(1+\gamma\right).
\end{equation}

\subsubsection{Outage Probability}
The OP is defined as the probability that the instantaneous channel capacity falls below a target rate $R_0$, i.e.,
\begin{equation}
    P_{\mathrm{out}} \triangleq \Pr\!\left( \log_2(1+\gamma) \leq R_0 \right).
    \label{eq:OP-def}
\end{equation}
Substituting \eqref{eq:SNR-def} into \eqref{eq:OP-def}, the OP can be equivalently expressed in terms of the cascaded channel gain $G_0$ as
\begin{equation}
    P_{\mathrm{out}}
    = \Pr\!\left(G_0 < \tilde{R} \right)
    = F_{G_0}\!\left( \tilde{R} \right),
    \label{eq:OP-final}
\end{equation}
where $F_{G_0}(\cdot)$ denotes the CDF of $G_0$, and
\begin{equation}
    \tilde{R} \triangleq \frac{2^{R_0}-1}{\bar{\gamma} L_\text{f} L_\text{u}}.
    \label{eq:SNR-rate}
\end{equation}

\subsubsection{Ergodic Capacity}
The EC, in bits/s/Hz, is defined as the expectation of the instantaneous channel capacity, namely,
\begin{align}
    \bar{C}
    &\triangleq \mathbb{E}\!\left[\log_2(1+\gamma)\right] \nonumber\\
    &= \int_{0}^{\infty}
    \log_2\!\left(1+\bar{\gamma}\, L_\text{f} L_\text{u}\, g\right)
    f_{G_0}(g)\, \mathrm{d}g,
    \label{eq:C-def}
\end{align}
where $f_{G_0}(\cdot)$ denotes the probability density function (PDF) of $G_0$.

% \vspace{1mm}
% \noindent
From \eqref{eq:OP-final} and \eqref{eq:C-def}, it is evident that evaluating the performance of the considered FRIS-assisted wireless communication system requires a complete statistical characterization of $G_0$, namely, the derivation of its PDF and CDF. However, obtaining these statistics is highly nontrivial, and to the best of our knowledge, no exact closed-form results have been previously reported. In the next section, we address this gap by deriving the first exact closed-form expressions for the PDF and CDF of $G_0$.

% Let $x$ have unit power and let $z \sim \mathcal{CN}(0,N_0)$ denote the additive noise.  
% Define the average transmit SNR as $\bar{\gamma}=P/N_0$.  
% Using~\eqref{eq:A-theta}, the received signal is
% \begin{equation}
% y = \sqrt{P L_\text{f} L_\text{u}}\;
% \mathbf{g}_\text{u}^\text{H}\mathbf{A}(\boldsymbol{\theta})\mathbf{g}_\text{f}\,x + z ,
% \end{equation}
% which yields the instantaneous SNR
% \begin{equation}
% \gamma = \bar{\gamma}\,L_\text{f} L_\text{u}\,G_0,
% \qquad
% G_0 = \big|\mathbf{g}_\text{u}^\text{H}\mathbf{A}(\boldsymbol{\theta})\mathbf{g}_\text{f}\big|^2 .
% \end{equation}

% For a target spectral efficiency $R_0$, an outage occurs when
% \[
% \log_2(1+\gamma) < R_0.
% \]
% This is equivalent to $G_0$ falling below a threshold, leading to
% \begin{equation}
% P_{\mathrm{out}}
% = \Pr\!\left(G_0 < \frac{2^{R_0}-1}{\bar{\gamma}L_\text{f} L_\text{u}}\right)
% = F_{G_0}\!\left(\frac{2^{R_0}-1}{\bar{\gamma}L_\text{f} L_\text{u}}\right),
% \label{eq:outage_ls}
% \end{equation}
% where $F_{G_0}(\cdot)$ is the cumulative distribution function (CDF) of $G_0$.

% Equation~\eqref{eq:outage_ls} shows that the large-scale fading terms $L_\text{f}$ and $L_\text{u}$ affect the outage probability only by \emph{scaling the threshold} of the CDF.  
% They do not modify the underlying distribution of $G_0$, which depends solely on the spatial correlation, the active-set geometry, and the phase configuration.

\section{Statistical Characterization}
% \section{Total Channel Gain Statistics}
\label{sec: Channel Gain Statistics}

In this section, we derive exact, closed-form expressions for the PDF and CDF of the end-to-end channel gain $G_0$. 
These derivations are attained in the next Theorem and Corollaries.

\subsection{Main Results}

\begin{theorem}
\label{eq: Theorem}
The PDF and CDF of $G_0$ can be written as a finite linear combination of {$K$-distributions} given, respectively, by
% as finite mixtures of K-distributions, given by
\begin{align}
f_{G_0}(g)
&= \sum_{i=1}^{q}\sum_{k=1}^{m_i} c_{i,k}\, f_{K}\!\big(k,\lambda_i,g\big),
\label{eq:PDF_G0}\\
F_{G_0}(g)
&= \sum_{i=1}^{q}\sum_{k=1}^{m_i} c_{i,k}\, F_{K}\!\big(k,\lambda_i,g\big),
\label{eq:CDF_G0}
\end{align}
where $f_K(\cdot)$ and $F_K(\cdot)$ denote the PDF and CDF of a {$K$–distribution} with shape parameter $k$ and scale parameter $\lambda_i$, given by
\begin{align}
f_{K}\!\big(k,\lambda_i,g\big)
&= \frac{2\,\lambda_i^{-\frac{k+1}{2}}}{\Gamma(k)}\,
   g^{\frac{k-1}{2}}
   K_{k-1}\!\left(2\sqrt{\frac{g}{\lambda_i}}\right)
\label{eq:PDF_K}\\
F_{K}\!\big(k,\lambda_i,g\big)
&= 1 - \frac{2}{\Gamma(k)}
   \left(\frac{g}{\lambda_i}\right)^{\frac{k}{2}}
   K_{k}\!\left(2\sqrt{\frac{g}{\lambda_i}}\right).
\label{eq:CDF_K}
\end{align}

Furthermore, let $\{\lambda_i\}_{i=1}^{q}$ denote the $q$ distinct \emph{strictly positive} eigenvalues of
$\mathbf{C}=\mathbf{A}\mathbf{A}^{\mathsf{H}}$, with corresponding multiplicities $\{m_i\}_{i=1}^{q}$.
General multiplicities are allowed, i.e., $m_i>1$ is permitted. These multiplicities satisfy
$\sum_{i=1}^{q} m_i = r$, where $r \triangleq \operatorname{rank}(\mathbf{C}) \le M_{\mathrm{on}}$.
Zero eigenvalues are excluded, as they do not contribute to the mixture representation.

% Furthermore, let $r \triangleq \operatorname{rank}(\mathbf{C}) \le M_{\mathrm{on}}$.
% If $\mathbf{C}$ is rank-deficient, it has $M_{\mathrm{on}}-r$ zero eigenvalues.
% Let $\{\lambda_i\}_{i=1}^{q}$ denote the $q$ distinct \emph{strictly positive} eigenvalues of $\mathbf{C}$,
% with multiplicities $\{m_i\}_{i=1}^{q}$ (general multiplicities), such that $\sum_{i=1}^{q} m_i = r$.
% Zero eigenvalues are excluded since they do not contribute to the mixture representation.}

The coefficients $c_{i,k}$ satisfy  
$\sum_{i=1}^{q}\sum_{k=1}^{m_i} c_{i,k}=1$ and are given by
\begin{align}
    \label{eq: coeff cik}
    c_{i,k} = \frac{1}{(m_i-k)!}\, \frac{d^{\,m_i-k}}{du^{\,m_i-k}} H_i(u)\Big|_{u=0},
\end{align}
where 
\begin{align}
    \label{}
    H_i(u)= \prod_{\substack{j=1\\ j\neq i}}^{q} \left( \frac{\lambda_j}{\lambda_i} u + 1 -  \frac{\lambda_j}{\lambda_i}\right)^{-m_j}.
\end{align}

\end{theorem}

\begin{proof}
see Appendix~\ref{app: Theorem}.    
\end{proof}

The Theorem also leads to the following corollaries.

\begin{corollary}
\label{sec: corollary 1}
If $m_i=1$ for every $i$ (i.e., all eigenvalues of $\mathbf{C}$ are simple and the
Laplace transform contains only first–order poles), the PDF and CDF in
\eqref{eq:PDF_G0} and \eqref{eq:CDF_G0} reduce, respectively, to
\begin{align}
f_{G_0}(g)
&= \sum_{i=1}^{q} c_{i,1}\, f_{K}\!\big(1,\lambda_i,g\big),
\label{eq:PDF_G0-simple}\\[1mm]
F_{G_0}(g)
&= \sum_{i=1}^{q} c_{i,1}\, F_{K}\!\big(1,\lambda_i,g\big),
\label{eq:CDF_G0-simple}
\end{align}
where $f_K(\cdot)$ and $F_K(\cdot)$ denote the PDF and CDF of a
$K$–distribution with shape parameter $1$ and scale parameter $\lambda_i$, given by
\begin{align}
f_{K}\!\big(1,\lambda_i,g\big)
&= \frac{2}{\lambda_i}\,
   K_{0}\!\left(2\sqrt{\frac{g}{\lambda_i}}\right),
\label{eq:PDFsimple-final} \\[1mm]
F_{K}\!\big(1,\lambda_i,g\big)
&= 1 -
   2\sqrt{\frac{g}{\lambda_i}}\,
   K_{1}\!\left(2\sqrt{\frac{g}{\lambda_i}}\right).
\label{eq:CDFsimple-final}
\end{align}
In addition, the coefficients satisfy $\sum_{i=1}^{q} c_{i,1}=1$ and are given by
\begin{align}
c_{i,1}
= \prod_{\substack{j=1 \\ j\neq i}}^{q}
  \frac{\lambda_i}{\lambda_i-\lambda_j}.
\label{eq:ci1-simple}
\end{align}
\end{corollary}

\begin{proof}
see Appendix~\ref{app: corollary 1}.    
\end{proof}

\begin{corollary}
\label{sec: corollary 2}
If $q=1$ (i.e., $\mathbf{C}$ has a single distinct eigenvalue $\lambda_1$ with
multiplicity $m_1$), then the quadratic form $T = \mathbf{g}_\text{u}^\text{H}\mathbf{C}\mathbf{g}_\text{u}$
follows a Gamma distribution with shape $m_1$ and scale $\lambda_1$, and the
PDF and CDF of $G_0$ reduce, respectively, to
\begin{align}
    f_{G_0}(g)
    &= \frac{2}{\Gamma(m_1)}\,
       \lambda_1^{-\frac{m_1+1}{2}}\,
       g^{\frac{m_1-1}{2}}\,
       K_{m_1-1}\!\left(2\sqrt{\frac{g}{\lambda_1}}\right) 
       \label{eq:PDF-cor2-final} \\
    F_{G_0}(g)
    &= 1 -
       \frac{2}{\Gamma(m_1)}\,
       \lambda_1^{-\frac{m_1}{2}}\,
       g^{\frac{m_1}{2}}\,
       K_{m_1}\!\left(2\sqrt{\frac{g}{\lambda_1}}\right),
       \label{eq:CDF-cor2-final}
\end{align}
which correspond to the classical $K$-distribution with
shape parameter $m_1$ and scale parameter $\lambda_1$.
\end{corollary}

\begin{proof}
see Appendix~\ref{app: corollary 2}.  
\end{proof}

\begin{corollary}
\label{sec: corollary 3}
If the FRIS elements are intrinsically uncorrelated (i.e., $\mathbf{R}_{\mathcal{S}}=\mathbf{I}_{M_{\mathrm{on}}}$), 
the PDF and CDF of $G_0$ reduce, respectively, to
\begin{align}
f_{G_0}(g)
&=
\frac{2}{\Gamma(M_{\mathrm{on}})}
\, g^{\frac{M_{\mathrm{on}}-1}{2}}\,
K_{M_{\mathrm{on}}-1}\!\left(2\sqrt{g}\right)
\label{eq:PDF_G0_independent}
\\
F_{G_0}(g)
&=
1 -
\frac{2}{\Gamma(M_{\mathrm{on}})}
\, g^{\frac{M_{\mathrm{on}}}{2}}\,
K_{M_{\mathrm{on}}}\!\left(2\sqrt{g}\right).
\label{eq:CDF_G0_independent}
\end{align}
\end{corollary}

\begin{proof}
In the uncorrelated case, 
$\mathbf{R}_{\mathcal{S}}=\mathbf{I}_{M_{\mathrm{on}}}$, and therefore
$\mathbf{A}
= \mathbf{R}_{\mathcal{S}}^{1/2}\boldsymbol{\Phi}\mathbf{R}_{\mathcal{S}}^{1/2}
= \boldsymbol{\Phi}$. 
Since $\boldsymbol{\Phi}$ is unitary, 
$\mathbf{C}=\mathbf{A}\mathbf{A}^\text{H}
= \boldsymbol{\Phi}\boldsymbol{\Phi}^\text{H}
= \mathbf{I}_{M_{\mathrm{on}}}$. 
Thus, $\mathbf{C}$ has one distinct eigenvalue 
$\lambda_1 = 1$ with multiplicity $m_1 = M_{\mathrm{on}}$, meaning that all
eigenvalues are equal and appear $M_{\mathrm{on}}$ times.
This setting corresponds precisely to the equal-eigenvalue case addressed in 
Corollary~\ref{sec: corollary 2}. Hence, the PDF and CDF of $G_0$ follow directly 
from the general expressions in that corollary by substituting 
$\lambda_1 = 1$ and $m_1 = M_{\mathrm{on}}$, which yields
\eqref{eq:PDF_G0_independent} and \eqref{eq:CDF_G0_independent}. 
This completes the proof.
\end{proof}

It is important to emphasize that, although the intrinsically uncorrelated case (Corollary~\ref{sec: corollary 3}) arises as a special instance of the general correlated FRIS framework, this particular regime has been well studied in the literature. In fact, the PDF expression in \eqref{eq:PDF_G0_independent} exactly coincides with the one reported in \cite[eq. (7)]{Heliot2024}. This perfect agreement not only validates our analytical framework in the uncorrelated setting but also confirms that the proposed general solution seamlessly recovers all previously known special cases.

\subsection{Physical Interpretation and Operating Regimes}

The Theorem and its three Corollaries admit clear physical interpretations within the FRIS framework. Together, they characterize four distinct operating regimes of the FRIS-assisted channel, each governed by the eigenvalue structure of $\mathbf{C}=\mathbf{A}\mathbf{A}^{\mathrm H}$ and, consequently, by the underlying fluid-induced spatial configuration:

\begin{enumerate}
    \item \textit{Theorem -- General correlation (general multiplicities):}
This theorem provides the most general statistical characterization of the FRIS-assisted cascaded channel by allowing the composite matrix $\mathbf{C}$ to exhibit an arbitrary eigenvalue spectrum, including repeated eigenvalues with non-unit multiplicities $\left\{m_i\right\}_{i=1}^q$. Apart from positive semidefiniteness, no structural assumptions are imposed on the spatial correlation, making the result applicable to virtually all practical FRIS configurations arising from fluidic repositioning and active-set selection over a finite aperture.

From a physical standpoint, the eigenstructure of $\mathbf{C}$ yields a modal decomposition of the FRIS-induced spatial correlation. Each distinct eigenvalue $\lambda_i$ represents a correlation-supported spatial mode,\footnote{A \emph{spatial mode} is an eigen-direction of $\mathbf{C}=\mathbf{A}\mathbf{A}^{\mathrm H}$ corresponding to an independent spatial degree of freedom of the FRIS-assisted channel, with its eigenvalue quantifying the mode’s strength.} while its multiplicity $m_i$ quantifies the number of statistically equivalent degrees of freedom sharing the same modal gain. Repeated eigenvalues, therefore, reflect degeneracies arising from geometric symmetries, regular activation patterns, or clusters of elements experiencing similar propagation conditions. Such degeneracies are generally unavoidable in realistic deployments and must be explicitly accounted for. Consequently, this regime bridges the fully decorrelated and fully correlated extremes, capturing partial correlation suppression via fluidic repositioning through mixed eigenvalue multiplicities.

    % This regime represents the most general and practically relevant operating
    % condition, encompassing \emph{any} FRIS geometry or activation pattern. It
    % captures realistic non-uniform element spacing produced by fluidic
    % repositioning, heterogeneous element-to-UE/BS distances, partial (rather
    % than complete) correlation reduction, and spatially structured correlation
    % profiles (e.g., edge versus center elements). In real deployments, the
    % eigenvalues of $\mathbf{C}$ naturally exhibit arbitrary and often mixed
    % multiplicities, placing this scenario between the fully decorrelated and
    % fully correlated extremes.

    \item \textit{Corollary~\ref{sec: corollary 1} -- Effectively decorrelated (all eigenvalues simple):}
This corollary characterizes the operating regime in which the composite matrix $\mathbf{C}$ has a \emph{simple spectrum}, i.e., all eigenvalues are distinct and have unit multiplicity ($m_i=1$ for all $i$). Although the underlying channel vectors may still be spatially correlated, the absence of repeated eigenvalues implies that no correlation mode is shared by multiple statistically equivalent degrees of freedom.

Physically, this regime corresponds to a FRIS configuration in which fluidic repositioning and active-set selection sufficiently perturb the geometry so that the correlation structure decomposes into distinct spatial modes, each contributing through a unique effective gain. This does not require a diagonal correlation matrix or independent and identically distributed (i.i.d.) channel entries; it only requires the absence of degenerate eigenspaces. Accordingly, ``effective decorrelation'' is a statistical and operational notion rather than a strictly physical one: the channel remains spatially correlated, but this correlation is redistributed across distinct modes, yielding performance behavior that closely approximates that of uncorrelated channels while preserving a correlated physical structure.
    % This regime arises when the FRIS repositions its active elements such that
    % all spatial modes become distinct ($m_i=1$ for every $i$), thereby
    % suppressing residual spatial correlation. Although the elements remain
    % decorrelated, their associated spatial modes generally possess unequal
    % power levels (unequal eigenvalues). As a result, this regime yields strong,
    % but not necessarily optimal, diversity performance. It corresponds to the
    % best configuration achievable purely through fluidic geometry optimization.

    \item \textit{Corollary~\ref{sec: corollary 2} -- Fully correlated (equal eigenvalues):}
This corollary characterizes the extreme operating regime in which the composite matrix $\mathbf{C}$ admits a single distinct eigenvalue, i.e., $q=1$ with eigenvalue $\lambda_1$ of multiplicity $m_1$. In this case, all spatial degrees of freedom collapse into a single correlation-supported mode, reflecting a fully correlated FRIS configuration.
% and the quadratic form $T=\mathbf{g}_\text{u}^{\mathrm H}\mathbf{C}\mathbf{g}_\text{u}$ follows a Gamma distribution with shape parameter $m_1$ and scale parameter $\lambda_1$. 
% Consequently, the cascaded channel gain $G_0$ follows the classical $K$-distribution given in \eqref{eq:PDF-cor2-final}–\eqref{eq:CDF-cor2-final}.

From a physical perspective, the presence of a single eigenvalue implies maximal spatial correlation across the FRIS, as all active elements contribute through statistically indistinguishable propagation paths with identical correlation-induced power scaling. Such behavior arises when the FRIS geometry fails to create distinct spatial modes—for instance, when active elements experience nearly identical geometric conditions, are closely spaced relative to the wavelength, or when fluidic repositioning is restricted to a region smaller than the spatial coherence length.
As a result, fluidic repositioning offers no effective decorrelation gain in this regime, and performance is entirely governed by the single dominant spatial mode. This case, therefore, constitutes a natural performance lower bound, against which the benefits of partial or effective decorrelation achieved in the other regimes can be quantitatively assessed.
    
    % This regime represents the opposite extreme, in which all active elements
    % behave as if fully correlated, producing a single distinct eigenvalue with
    % multiplicity equal to the rank of $\mathbf{C}$. Such conditions arise under
    % highly correlated or rigid layouts---e.g., contiguous non-fluid RIS blocks
    % or poorly optimized FRIS activation patterns. It produces the most severe
    % fading conditions and therefore constitutes a performance lower bound
    % for FRIS operation.

    \item \textit{Corollary~\ref{sec: corollary 3} -- Intrinsically uncorrelated:}
This corollary describes the idealized operating regime in which the FRIS-assisted channel is intrinsically uncorrelated, meaning that the underlying channel coefficients are statistically independent by construction. In this case, the composite matrix reduces to $\mathbf{C}=\mathbf{I}_{M_{\mathrm{on}}}$, yielding a single eigenvalue $\lambda_1=1$ with multiplicity $M_{\mathrm{on}}$, and hence all spatial degrees of freedom are independent and identically distributed.

From a physical standpoint, intrinsic uncorrelation corresponds to a rich-scattering environment in which each active FRIS element experiences independent propagation conditions, with negligible mutual coupling and sufficiently large inter-element separation relative to the wavelength and coherence distance. This regime is not induced by fluidic repositioning, but rather assumed as an ideal baseline channel model.
Statistically, intrinsic uncorrelation constitutes a limiting case of the general framework and should not be confused with effective decorrelation. While both may exhibit similar performance trends, intrinsic uncorrelation implies a complete absence of spatial correlation at the channel level, whereas effective decorrelation arises from redistributing correlation across distinct spatial modes via FRIS geometry. Accordingly, this regime primarily serves as an analytical reference and upper-bound performance benchmark, rather than a practically attainable outcome of fluidic repositioning.
    
    % This special regime occurs when the FRIS elements are intrinsically
    % uncorrelated (i.e., $\mathbf{R}_{\mathcal{S}}=\mathbf{I}_{M_{\mathrm{on}}}$) due to their physical configuration or inherent properties.
    % Here, all elements experience identical large-scale conditions, no mutual
    % coupling, and fully isotropic responses. The composite matrix collapses to $\mathbf{C}=\mathbf{I}_{M_{\mathrm{on}}}$, yielding a single eigenvalue $\lambda_1=1$ with
    % multiplicity $M_{\mathrm{on}}$. This corresponds to the ideal independent
    % and identically distributed fading regime and yields the statistically
    % best performance for FRIS, matching a $K$-distribution with shape parameter
    % $M_{\mathrm{on}}$.

\end{enumerate}

Collectively, these four regimes span the full spectrum of statistical behaviors enabled by FRIS operation, ranging from the fully correlated case (worst case), through general mixed-correlation scenarios, to the effectively decorrelated regime, and ultimately to the intrinsically uncorrelated i.i.d.\ case (best case). This layered interpretation highlights how fluidic reconfiguration allows an FRIS to traverse these regimes by dynamically adjusting element placement and, consequently, the induced spatial correlation structure.

\subsection{Applicability to Conventional RIS Systems}

Although the analytical framework developed in this work is motivated by the fluid-reconfigurable architecture of FRIS, the results of the theorem and its corollaries apply directly to conventional RIS systems. When $M_{\mathrm{on}} = M$ (i.e., the FRIS and the RIS activate the same number of elements) and the element locations are fixed, the proposed framework naturally reduces to the standard RIS setting. Importantly, the formulations presented in the theorem and corollaries accommodate \emph{spatially correlated} RIS configurations---including fully correlated, partially correlated, effectively decorrelated, and uncorrelated cases---which commonly arise in practice due to finite inter-element spacing, mutual electromagnetic coupling, substrate interactions, and embedding effects.
In this scenario, inserting a general RIS spatial correlation matrix into \eqref{eq:A-theta} yields a matrix $\mathbf{A}$ whose composite matrix $\mathbf{C}$ exhibits an arbitrary eigenvalue structure. As a result, conventional correlated (or uncorrelated) RIS channels emerge as special cases of the broader FRIS analytical model. Therefore, the theorem and its corollaries also provide a complete and exact characterization of the cascaded channel in conventional RIS systems.
To the best of our knowledge, such exact closed-form expressions for correlated RIS channels---expressed as finite linear combinations of $K$-distributions---have not been previously reported in the RIS literature, which has predominantly relied on Gamma, generalized-$K$, or asymptotic/CLT-based approximations \cite{10275022,11152354,10683015,9424713,10164200}.

% \textit{Remark (Uniformly Random Phases):}
% When the FRIS phase shifts $\theta_\ell$ are independently and uniformly
% distributed over $[0,2\pi)$, the first-order statistics of the cascaded channel
% gain $G_0$---namely, its PDF and CDF as derived in the Theorem and corresponding
% corollaries---remain identical to those obtained under deterministic phase
% configurations.
% This invariance stems from the intrinsic structure of the FRIS-assisted channel.
% Under circularly symmetric complex Gaussian fading, multiplying any entry of
% $\mathbf{g}_\text{u}$ or $\mathbf{g}_\text{f}$ by a unit-modulus complex exponential
% $e^{\ii\theta_\ell}$, with $\theta_\ell \sim \mathcal{U}[0,2\pi)$, does not alter
% the joint distribution of these vectors. Consequently, the distribution of the
% quadratic form $
% G_0 = \big|\mathbf{g}_\text{u}^\text{H}\mathbf{A}\mathbf{g}_\text{f}\big|^{2}$
% is invariant to uniformly random phase rotations introduced through
% $\boldsymbol{\Phi}$.
% As a result, deterministic and uniformly random phase configurations yield
% identical PDF, CDF, OP, and EC expressions. This theoretical insight is further
% corroborated by the numerical results in Section~\ref{sec: Numerical Results}.

\section{Performance Analysis}
\label{sec: Performance Analysis}

In this section, we assess the performance of the FRIS-aided wireless
communication system described in Section~\ref{sec: System Model} in an exact
and asymptotic manner. 
Specifically, we derive the exact and
asymptotic expressions for the system's OP for each of the operating regimes, as well as exact expressions for the EC.

\subsection{Outage Probability}

\subsubsection{General-Multiplicity Case}

For the general-multiplicity case, the exact OP follows directly by substituting \eqref{eq:CDF_G0} into \eqref{eq:OP-final}, yielding
\begin{align}
\label{eq:OP-final-general}
P_{\mathrm{out}}
    = \sum_{i=1}^{q} \sum_{k=1}^{m_i}
      c_{i,k}\,
      P_{\mathrm{out}}^{K}\!\left(k,\lambda_i,\tilde{R}\right),
\end{align}
where $P_{\mathrm{out}}^{K}(k,\lambda_i,\tilde{R})$ denotes the OP of a single
$K$-distributed random variable with shape parameter $k$ and scale parameter
$\lambda_i$, given by
\begin{align}
\label{eq:OP-K-general}
P_{\mathrm{out}}^{K}\!\left(k,\lambda_i,\tilde{R}\right)
    = 1 - \frac{2}{\Gamma(k)}
      \left(\frac{\tilde{R}}{\lambda_i}\right)^{\frac{k}{2}}
      K_{k}\!\left(2\sqrt{\frac{\tilde{R}}{\lambda_i}}\right).
\end{align}
Hence, the exact OP is expressed as a weighted linear combination of the OPs associated with the individual $K$-distributed components induced by the eigenvalue multiplicities.

We now investigate the high-SNR regime, i.e., $\bar{\gamma} \to \infty$, or equivalently $\tilde{R} \to 0$. In the subsequent derivations, we retain only the dominant (leading-order) terms that determine the asymptotic scaling behavior at high SNR. To this end, we employ the small-argument expansion of the modified Bessel function of the second kind, $K_\nu(z)$, as $z \to 0$, for integer order $\nu \in \mathbb{N}$~\cite{NIST_DLMF}:
\begin{align}
\label{eq:K1-small}
K_1(z) &\simeq \frac{1}{z} + \frac{z}{2} \log\!\left(\frac{z}{2}\right), 
\\ \label{eq:Kk-small}
K_k(z) &\simeq \frac{1}{2}(k-1)! \left(\frac{2}{z}\right)^{k}
-\frac{1}{2}(k-2)! \left(\frac{z}{2}\right)^{2-k}, \quad k \geq 2.
\end{align}

Substituting \eqref{eq:K1-small} into \eqref{eq:OP-K-general} yields
\begin{align}
P_{\mathrm{out}}^{K}(1,\lambda_i,\tilde{R})
&\simeq
-\frac{\tilde{R}}{\lambda_i}
\log\!\left(\frac{\tilde{R}}{\lambda_i}\right),
\label{eq:asympt-PoutK-k1}
\end{align}
whereas substituting \eqref{eq:Kk-small} into \eqref{eq:OP-K-general} gives, for $k\geq2$,
\begin{equation}
P_{\mathrm{out}}^{K}(k,\lambda_i,\tilde{R})
\simeq
\frac{1}{k-1}
\left(\frac{\tilde{R}}{\lambda_i}\right).
\label{eq:asympt-PoutK-kge2}
\end{equation}

Using \eqref{eq:SNR-rate}, \eqref{eq:asympt-PoutK-k1}, and \eqref{eq:asympt-PoutK-kge2} in \eqref{eq:OP-final-general}, and after straightforward algebraic manipulations, the asymptotic OP can be expressed as
\begin{align}
\nonumber P_{\mathrm{out}}
&\simeq
\frac{2^{R_0}-1}{L_\text{f} L_\text{u} \bar{\gamma}} \\
& \times
\Bigg[
S_1 \log\!\left( \bar{\gamma} \right)
- S_1 \log\!\left( \frac{2^{R_0}-1}{L_\text{f} L_\text{u}} \right)
+ S_2 + S_3
\Bigg],
\label{eq:Pout-highSNR general}
\end{align}
where
\begin{subequations}\label{eq:S_terms}
\begin{align}
S_1 &= \sum_{i=1}^{q} \frac{c_{i,1}}{\lambda_i}, \label{eq:S1}\\
S_2 &= \sum_{i=1}^{q} \frac{c_{i,1}\log(\lambda_i)}{\lambda_i}, \label{eq:S2}\\
S_3 &= \sum_{i=1}^{q}\sum_{k=2}^{m_i}
       \frac{c_{i,k}}{(k-1)\,\lambda_i}. \label{eq:S3}
\end{align}
\end{subequations}

% Finally, the diversity order is obtained as
% \begin{equation}
% d
% = -\lim_{\bar{\gamma}\to\infty}
% \frac{\log (P_{\mathrm{out}})}{\log (\bar{\gamma})}
% = 1.
% \end{equation}

\subsubsection{Simple-Eigenvalue Case}

For the simple-eigenvalue case ($m_i = 1$ for all $i$), each eigenvalue contributes a single $K$-distributed component. Substituting \eqref{eq:CDF_G0-simple} into \eqref{eq:OP-final} yields the exact OP as
\begin{align}
\label{eq:OP-simple-mixture}
P_{\mathrm{out}}
    = \sum_{i=1}^{q}
      c_{i,1}\,
      P_{\mathrm{out}}^{K}\!\left(1,\lambda_i,\tilde{R}\right),
\end{align}
where the OP of a single $K$-distributed random variable with shape parameter $k=1$ and scale parameter $\lambda_i$ is given by
\begin{align}
\label{eq:OP-simple-K}
P_{\mathrm{out}}^{K}\!\left(1,\lambda_i,\tilde{R}\right)
    = 1 -
      2\sqrt{\frac{\tilde{R}}{\lambda_i}}\,
      K_{1}\!\left(2\sqrt{\frac{\tilde{R}}{\lambda_i}}\right).
\end{align}

By substituting \eqref{eq:SNR-rate} and the small-argument expansion in \eqref{eq:K1-small} into \eqref{eq:OP-simple-mixture}, and following steps analogous to those leading to \eqref{eq:Pout-highSNR general}, the asymptotic OP for the simple-eigenvalue case can be expressed as
\begin{align}
P_{\mathrm{out}}
&\simeq
\frac{2^{R_0}-1}{L_\text{f} L_\text{u} \bar{\gamma}}
\Bigg[
S_1 \log\!\left( \bar{\gamma} \right)
- S_1 \log\!\left( \frac{2^{R_0}-1}{L_\text{f} L_\text{u}} \right)
+ S_2
\Bigg],
\label{eq:Pout-highSNR-simple}
\end{align}
where $S_1$ and $S_2$ are defined in \eqref{eq:S1} and \eqref{eq:S2}, respectively. 
% Consequently, the diversity order is given by
% \begin{align}
% d
% &=
% -\lim_{\bar{\gamma}\to\infty}
% \frac{\log (P_{\mathrm{out}})}{\log (\bar{\gamma})}
% = 1.
% \end{align}

\subsubsection{Equal-Eigenvalue Case}

For the equal-eigenvalue case ($q=1$ and $m_1 = r$), the cascaded channel gain reduces to that of a single $K$-distributed random variable with shape parameter $m_1$ and scale parameter $\lambda_1$. Accordingly, substituting \eqref{eq:CDF-cor2-final} into \eqref{eq:OP-final} yields the exact OP as
\begin{align}
\label{eq:OP-final-equal}
P_{\mathrm{out}}
    = 1 -
      \frac{2}{\Gamma(m_1)}\,
      \lambda_1^{-\frac{m_1}{2}}\,
      \tilde{R}^{\frac{m_1}{2}}\,
      K_{m_1}\!\left(2\sqrt{\frac{\tilde{R}}{\lambda_1}}\right).
\end{align}

Replacing \eqref{eq:SNR-rate} and using the small-argument expansion in \eqref{eq:Kk-small}, the OP admits the following asymptotic form:
\begin{equation}
P_{\mathrm{out}}
\simeq
\frac{2^{R_0}-1}{
(m_1 - 1)\,\lambda_1\,L_\text{f} L_\text{u}\,\bar{\gamma}},
\label{eq:OP-equal-asymptotic-final}
\end{equation}
which holds for $m_1 \geq 2$. 
% Consequently, the diversity order is given by
% \begin{align}
% d
% &=
% -\lim_{\bar{\gamma}\to\infty}
% \frac{\log (P_{\mathrm{out}})}{\log (\bar{\gamma})}
% = 1.
% \end{align}

\subsubsection{Uncorrelated Case}

For the uncorrelated case, the composite matrix reduces to $\mathbf{C}=\mathbf{I}_{M_{\mathrm{on}}}$, and the cascaded channel gain $G_0$ follows a $K$-distribution with shape parameter $M_{\mathrm{on}}$ and unit scale parameter. Substituting \eqref{eq:CDF_G0_independent} into \eqref{eq:OP-final} yields the exact OP as
\begin{align}
\label{eq: OP final uncorrelated}
P_{\mathrm{out}}
    =
    1 -
    \frac{2}{\Gamma(M_{\mathrm{on}})}
    \, \tilde{R}^{\frac{M_{\mathrm{on}}}{2}}\,
    K_{M_{\mathrm{on}}}\!\left(2\sqrt{\tilde{R}}\right).
\end{align}

Replacing \eqref{eq:SNR-rate} and using the small-argument expansion in \eqref{eq:Kk-small} in \eqref{eq: OP final uncorrelated}, the OP admits the following asymptotic form:
\begin{equation}
P_{\mathrm{out}}
\simeq
\frac{2^{R_0}-1}{
(M_{\mathrm{on}} - 1)\,L_\text{f} L_\text{u}\,\bar{\gamma}},
\label{eq:OP-uncorr-asymptotic-final}
\end{equation}
which holds for $M_{\mathrm{on}} \geq 2$. 
% Thus, the diversity order is given by
% \begin{align}
% d
% &=
% -\lim_{\bar{\gamma}\to\infty}
% \frac{\log (P_{\mathrm{out}})}{\log (\bar{\gamma})}
% = 1.
% \end{align}

\subsubsection{Diversity Order}
% The diversity order is determined by the dominant term of the mixture distribution in the small-argument regime (as $R\to 0$).
For all four cases considered above, the diversity order $d$ is invariant and is given by
\begin{equation}
d
= -\lim_{\bar{\gamma}\to\infty}
\frac{\log (P_{\mathrm{out}})}{\log (\bar{\gamma})}
= 1.
\end{equation}

It is worth emphasizing that this result captures the \emph{true} diversity behavior of the FRIS-assisted cascaded channel under exact statistical characterization. Specifically, the diversity order is fundamentally equal to one, regardless of eigenvalue multiplicities, spatial correlation regimes, or fluidic reconfiguration, i.e., spatial correlation and the choice of $\boldsymbol{\Phi}$ only affect the mixture coefficients of the distribution, while the high-SNR exponent—and thus the diversity order---remains unchanged, yielding $d = 1$.
This finding refines the diversity behavior reported in~\cite[Corollary~1]{11154019Ghadi}, where the diversity order was inferred from a Gamma approximation and consequently tied to the shape parameter of the 
approximating distribution rather than to the exact asymptotic structure of the cascaded channel.

\subsection{Ergodic Capacity}

\subsubsection{General-Multiplicity Case}
For the general-multiplicity case, the exact EC is given by (see Appendix~\ref{app: ergodic capacity})
\begin{equation}
\bar C
= \frac{1}{\ln (2) }
\sum_{i=1}^{q}\sum_{k=1}^{m_i}
\frac{c_{i,k}}{\Gamma(k)}\,
G_{4,2}^{1,4}\!\left(
\bar{\gamma} L_\text{f} L_\text{u} \lambda_i\;\Bigg|\;
\begin{matrix}
1-k,\;0,\;1,\;1\\
1,\;0
\end{matrix}
\right).
\label{eq:C_total_meijerg}
\end{equation}

\subsubsection{Simple-Eigenvalue Case} For the simple-eigenvalue case, the exact EC is given by (see Appendix~\ref{app: ergodic capacity})
\begin{equation}
\bar C
=
\frac{1}{\ln (2)} \sum_{i=1}^{q}
c_{i,1}\,
G_{4,2}^{1,4}\!\left(
\bar{\gamma} L_\text{f} L_\text{u} \lambda_i\;\Bigg|\;
\begin{matrix}
0,\;0,\;1,\;1\\
1,\;0
\end{matrix}
\right).
\label{eq:C_total_meijerg_simple}
\end{equation}

\subsubsection{Equal-Eigenvalue Case}

For the equal-eigenvalue case, the exact EC is given by (see Appendix~\ref{app: ergodic capacity})
\begin{equation}
\label{eq:EC-final-equal}
\bar C
=
\frac{1}{\ln (2)\,\Gamma(m_1)}\,
G_{4,2}^{1,4}\!\left(
\bar{\gamma} L_\text{f} L_\text{u} \lambda_1\;\Bigg|\;
\begin{matrix}
1-m_1,\;0,\;1,\;1\\
1,\;0
\end{matrix}
\right).
\end{equation}

\subsubsection{Uncorrelated Case}
For the uncorrelated case, the exact EC is given by (see Appendix~\ref{app: ergodic capacity})
\begin{equation}
\label{eq:EC-final-uncorrelated}
\bar C
=
\frac{1}{\ln (2) \Gamma(M_{\rm on})}\,
G_{4,2}^{1,4}\!\left(
\bar{\gamma} L_\text{f} L_\text{u} \;\Bigg|\;
\begin{matrix}
1-M_{\rm on},\;0,\;1,\;1\\
1,\;0
\end{matrix}
\right).
\end{equation}

To the best of the author's knowledge, the closed-form expressions for the OP and EC derived in this work have not been previously reported in the FRIS or (correlated) RIS literature.

\section{Numerical Results and Discussion}
\label{sec: Numerical Results}

This section validates and corroborates all the proposed analytical expressions through extensive Monte Carlo (MC) simulations. In addition, as a theoretical accuracy benchmark, we adopt the Gamma approximation reported in \cite{11154019Ghadi}. Finally, to clearly demonstrate the performance advantages of the FRIS architecture, we include a comparative analysis against a conventional RIS system characterized by fixed, contiguous reflecting elements.

% and to evaluate the performance of the FRIS-assisted communication
% system in comparison with a conventional RIS-assisted system.
% The performance is assessed in terms of the outage probability, the
% probability density function of the effective cascaded channel gain, and the ergodic capacity.

The simulation parameters were set as follows.
A uniform planar array with $M_x=M_z=20$ elements was considered, with
$M_{\mathrm{on}}\in\{25,36\}$ active reflecting elements.
The normalized inter-element spacing was set to $d_w=0.15$ (dimensionless),
and the carrier wavelength was $\lambda_c=0.125\,\mathrm{m}$, corresponding to a carrier
frequency of approximately $2.4\,\mathrm{GHz}$.

The \ac{BS} and \ac{UE} distances to the FRIS/RIS are $20$\,m and $40$\,m, respectively.
% The \ac{BS} and the \ac{UE} were located at
% $\mathbf{b}_s = [-20,\,0]$\,m and $\mathbf{u}_e = [40,\,0]$\,m \changed{[Why are the distance in 2-dimensional coordinates? So, where is the origin [0,0]? Is this correct? 
% Is not just $\mathbf{b}_s=20$ m and $\mathbf{u}_e=40$ m?]},
% respectively.
Large-scale fading was modeled using a scalar pathloss formulation with
pathloss exponents $\alpha_\text{f} = \alpha_\text{u} = 2.1$, reference gains
$\rho_\text{f} = \rho_\text{u} = 10$, and reference distance $d_0 = 1$\,m.
% \paragraph*{Phase invariance of outage statistics}
% Although the Random phase configuration  explicitly averages over 
% independent random phase realizations and the Deterministic phase configuration uses a 
% single fixed phase configuration, both approaches produce identical 
% outage curves.
% This occurs because, under circularly symmetric complex Gaussian fading,
% the distribution of 
% $G_0 = |\mathbf{g}_\text{u}^\text{H}\mathbf{A}\mathbf{g}_\text{f}|^{2}$
% is invariant to the uniformly distributed phases in
% $\boldsymbol{\Phi}$.
% Therefore, once $\mathbf{C} = \mathbf{A}\mathbf{A}^\text{H}$ is normalized
% to have $\mathrm{tr}(\mathbf{C})=1$, the choice between fixed and random 
% phases does not affect the statistics of $G_0$.
% All differences observed between RIS and FRIS thus originate from the 
% spatial-correlation structure, not from the phase configuration.
% The FRIS activation pattern was generated using the
% correlation selection rule described in
% Section~\ref{activation}, with an initial threshold
% $\tau_{\text{init}} = 0.3$.
% After relaxation, the final correlation cap was
% $\tau_{\text{used}} = 0.421$, producing a maximum pairwise correlation
% of $|J_0 (\cdot)|_{\max} = 0.402$ among active FRIS elements.
% For the same number of active elements, the contiguous RIS block
% exhibited a significantly higher maximum correlation of
% $|J_0 (\cdot)|_{\max} = 0.79$.
For simulation purposes, the reflection phases ${\theta_\ell}$ are generated once by sampling from $\mathcal{U}[0,2\pi)$ and subsequently kept fixed. This choice is purely a numerical initialization strategy to obtain a generic phase configuration and should not be interpreted as modeling random or time-varying phase shifts.
The target rate was fixed at $R_0 = 0.1$\,bps/Hz.
% The SNR was swept from $0$\,dB to $80$\,dB, and the target rate was fixed at $R_0 = 0.1$\,bps/Hz.
All MC-based results were averaged over $3\times 10^{6}$ independent fading realizations to ensure statistical
convergence.

\subsection{Element Activation Patterns for RIS and FRIS}
To clearly differentiate the baseline RIS from the proposed FRIS, two distinct element-activation policies are considered.
In the conventional RIS, $M_r$ reflecting elements are activated as a contiguous block centered on the surface, with $M_r = 25$ and $M_r = 36$, as illustrated in Figs.~\ref{fig:elements}-(a) and \mbox{\ref{fig:elements}-(c)}, respectively.
In contrast, the FRIS activation patterns are generated using the correlation-based selection rule described in Section~\ref{activation}, with $M_{\mathrm{on}} = 25$ and $M_{\mathrm{on}} = 36$, as shown in Figs.~\ref{fig:elements}-(b) and~\ref{fig:elements}-(d), respectively. This approach spatially distributes the active elements to mitigate inter-element correlation, starting from an initial threshold $\tau_{\text{init}} = 0.3$.
After the relaxation procedure, the final correlation cap is $\tau_{\text{used}} = 0.421$, yielding a maximum pairwise correlation of $|J_0(\cdot)|_{\max} = 0.402$ among the active FRIS elements.
For the same number of active elements, the contiguous RIS configuration exhibits a significantly higher maximum correlation of $|J_0(\cdot)|_{\max} = 0.79$.
These activation patterns are used to generate the subsequent PDF, OP, and EC figures.

% the FRIS activates $M_{\mathrm{on}}$ elements according to the correlation-threshold criterion, which spatially distributes the active elements to mitigate inter-element correlation. 
% Figure~\ref{fig:elements} illustrates the resulting activation patterns for $M_r = M_{\mathrm{on}} \in \{25,36\}$.

\begin{figure}[t!]
\centering
\includegraphics[width=0.95\linewidth]{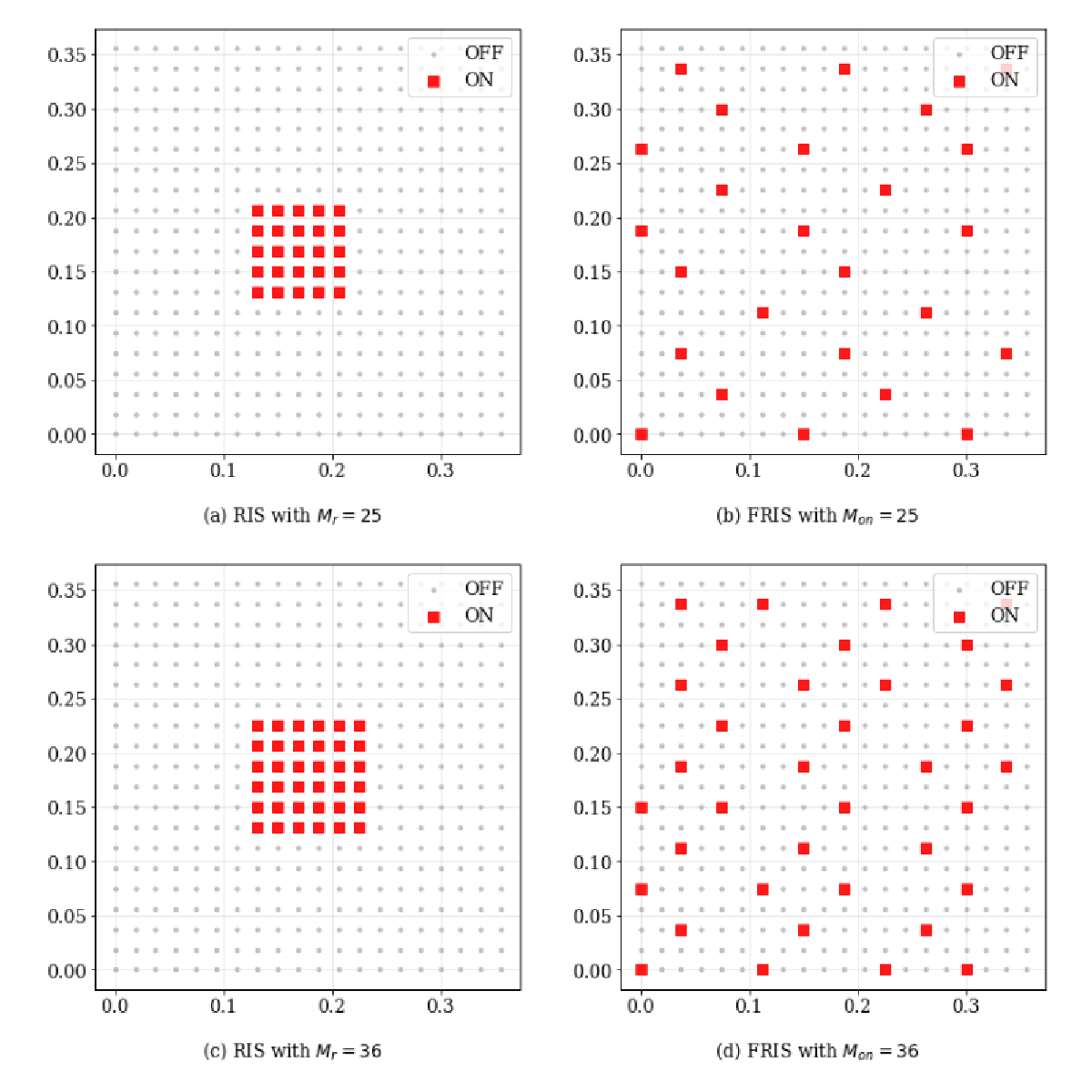}
\caption{Element activation patterns for RIS and FRIS.}
\label{fig:elements}
\end{figure}

\subsection{PDF Discussion}

Figs.~\ref{eq: PDF 25} and~\ref{eq: PDF 36} depict the PDF of the end-to-end cascaded channel gain, $G_0$, for different FRIS and RIS configurations. Specifically, Fig.~\ref{eq: PDF 25} considers $M_r = M_{\mathrm{on}} = 25$ active elements, whereas Fig.~\ref{eq: PDF 36} corresponds to $M_r = M_{\mathrm{on}} = 36$ active elements.
For both configurations, the proposed exact closed-form expressions exhibit excellent agreement with MC simulations, thereby validating the correctness and accuracy of the analytical derivations.
% In contrast, the Gamma-approximation PDF in \cite[eq.~(7)]{11154019Ghadi} exhibits a noticeable loss of accuracy in both the left and right tails of the distribution in both figures. These tail regions are of paramount importance for evaluating performance metrics in wireless communications (e.g., OP and ABER) and radar systems (e.g., probability of false alarm). While the Gamma approximation provides a reasonable fit around the main body (peak) of the distribution, it fails to accurately capture the tail behavior.
In contrast, although the Gamma-approximation PDF in \cite[eq.~(7)]{11154019Ghadi} provides a reasonable fit around the main body (peak) of the distribution, it fails to faithfully capture the behavior of both the left and right tails.\footnote{These tail regions are of paramount importance for evaluating performance metrics in wireless communications (e.g., OP and ABER) and radar systems (e.g., probability of false alarm), as such metrics are inherently dominated by rare but critical events.} This limitation stems from the fact that the Gamma approximation is constructed by matching only low-order moments, which is insufficient to accurately characterize extreme channel fluctuations~\cite{GarciaTVTris}.
This loss of tail accuracy is more pronounced in the conventional RIS scenario, where larger discrepancies---particularly in the left tail---are observed. This behavior can be attributed to stronger spatial correlation among contiguous reflecting elements and the resulting increased eigenvalue dispersion, effects that are effectively mitigated by the FRIS activation strategy.

% configurations, respectively. Monte Carlo simulations are compared against the proposed exact analytical PDF and a Gamma moment-matching approximation. For both configurations, the exact analytical expression exhibits excellent agreement with Monte Carlo results over the entire range of $g$, thereby validating the correctness of the derived Bessel--K model.

% A clear difference is observed in the accuracy of the Gamma approximation. In the FRIS case Fig.\ref{fris:pdf}, the Gamma-based PDF closely follows the Monte Carlo curve, with only minor deviations in the extreme low-gain region. This behavior is attributed to the reduced spatial correlation induced by the correlation-threshold selection, which leads to a more concentrated eigenvalue spectrum of $\mathbf{A}\mathbf{A}^\text{H}$. In contrast, for the conventional RIS configuration Fig.\ref{ris:pdf}, the Gamma approximation exhibits larger discrepancies, particularly in the left tail, due to stronger spatial correlation and increased eigenvalue dispersion. These results indicate that FRIS not only improves channel decorrelation but also enables simpler and more accurate statistical modeling.

\begin{figure}[t!]
 \centering
 \includegraphics[width=0.8\linewidth]{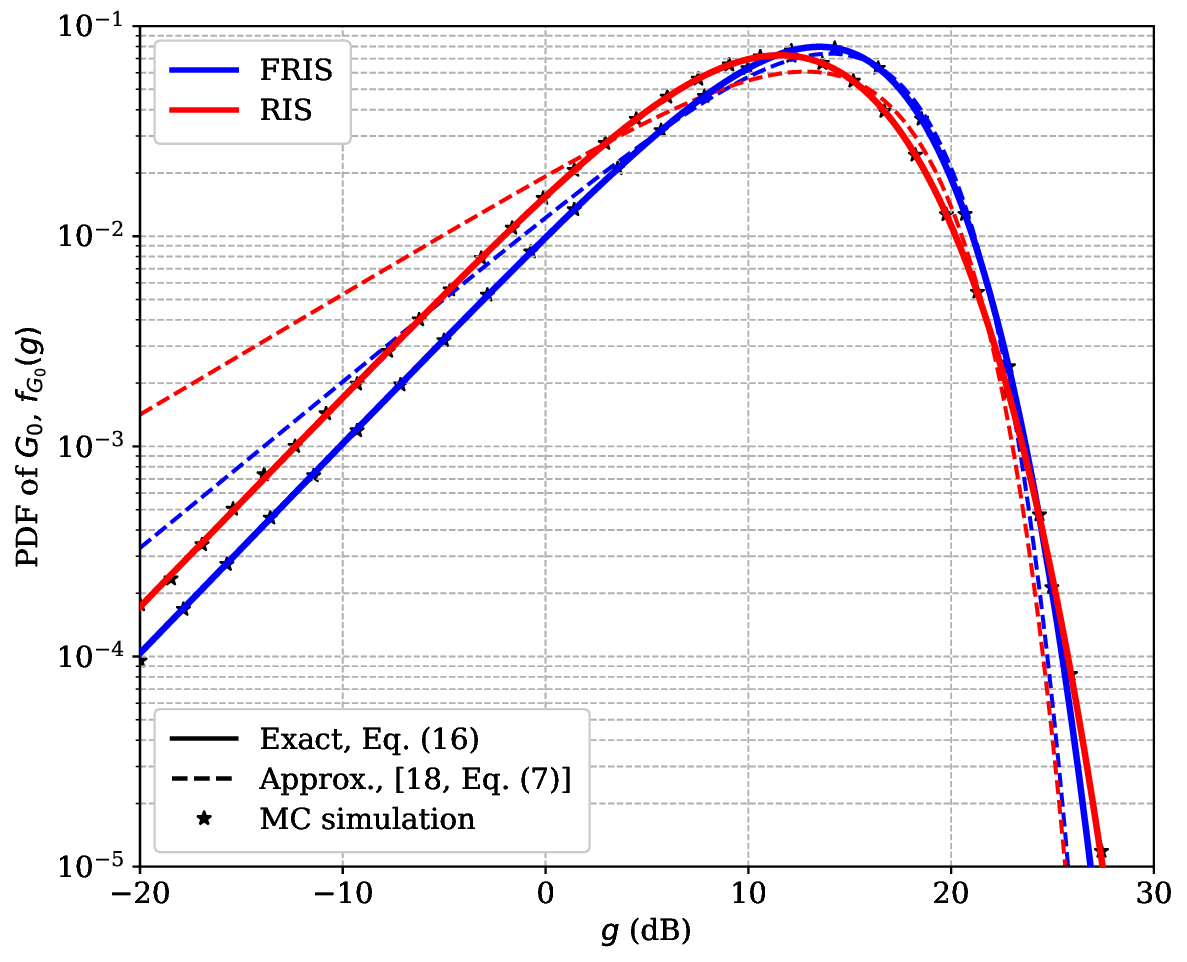}
 \caption{PDF of $G_0$ for different FRIS and RIS configurations with $ M_{\mathrm{on}} = 25$.}
 \label{eq: PDF 25}
 \end{figure}

 \begin{figure}[t!]
 \centering
 \includegraphics[width=0.8\linewidth]{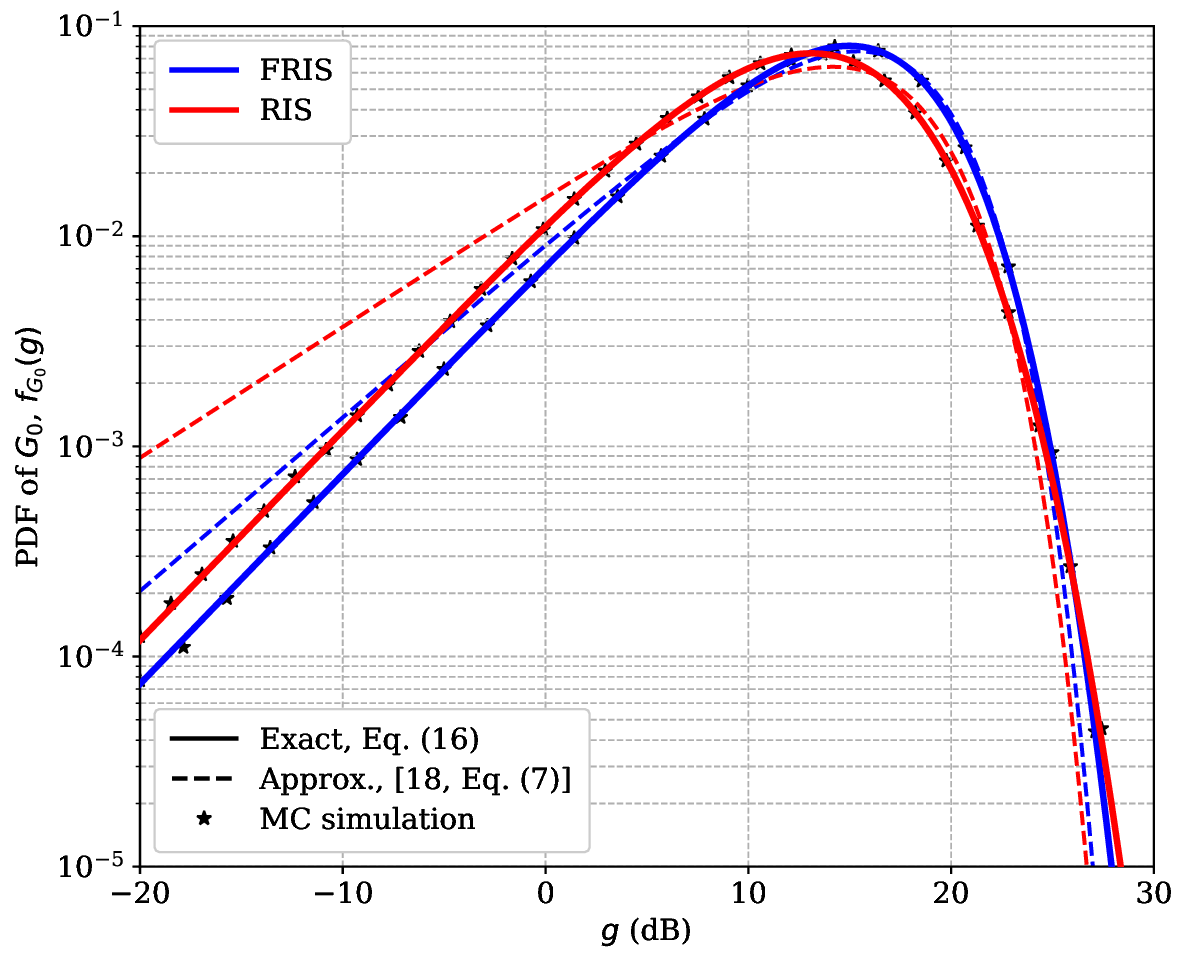}
 \caption{PDF of $G_0$ for different FRIS and RIS configurations with $M_{\mathrm{on}} = 36$.}
 \label{eq: PDF 36}
 \end{figure}

\subsection{Outage Probability Discussion}

Figs.~\ref{fig:asymptotic25} and~\ref{fig:asymptotic36} illustrate the OP as a function of the average SNR, $\bar{\gamma}$, for FRIS- and RIS-assisted systems with $M_r = M_{\mathrm{on}} = 25$ and $M_r = M_{\mathrm{on}} = 36$ reflecting elements, respectively. 
In both figures, the proposed exact analytical expressions exhibit perfect agreement with the MC simulation results, thereby corroborating the correctness and accuracy of the analytical derivations.
On the other hand, the Gamma approximation in \cite[eq.~(20)]{11154019Ghadi} exhibits noticeable accuracy discrepancies, yielding overly conservative estimates. These discrepancies become more pronounced in the high-SNR regime and are particularly severe for the conventional RIS scenario, where stronger spatial correlation among contiguous elements leads to more frequent deep fades. This behavior is consistent with the observations drawn from the left tails of the PDFs of $G_0$ in Figs.~\ref{eq: PDF 25} and~\ref{eq: PDF 36}, which dominate the OP performance at high SNR and are therefore critical to accurately capture.
Furthermore, it is observed that the FRIS-assisted system consistently outperforms the conventional RIS-assisted counterpart, yielding a moderate reduction in OP for a given average SNR. This performance advantage stems from the ability of FRIS to spatially distribute the active elements and thereby suppress inter-element correlation. Increasing the number of active elements---from $25$ to $36$ in this study---further enhances system performance by improving the effective channel strength and reducing the probability of outage.

The asymptotic curves provide an additional and relevant insight. Specifically, the \emph{true} diversity order is not affected by spatial correlation or by fluidic reconfiguration; it remains constant and equal to one, as rigorously confirmed by the proposed analytical expressions. In contrast, the diversity order reported in \cite[Corollary~1]{11154019Ghadi} depends explicitly on the correlation matrix. As a result, modifying the positions and correlation structure of the active elements alters the predicted diversity order, leading to different OP slopes, as observed in Figs.~\ref{fig:asymptotic25} and~\ref{fig:asymptotic36}.

\subsection{Ergodic Capacity Discussion}

Figs.~\ref{fig: EC 25} and~\ref{fig: EC 36} illustrate the EC as a function of the average SNR, $\bar{\gamma}$, for FRIS- and RIS-assisted systems with $M_r = M_{\mathrm{on}} = 25$ and $M_r = M_{\mathrm{on}} = 36$ reflecting elements, respectively. To compute the EC reported in \cite[eq.~(23)]{11154019Ghadi}, a numerical integration approach is employed.
The MC simulation results exhibit perfect agreement with the proposed exact formulations, thus ensuring their validity and accuracy. 
Conversely, in both figures, it is observed that the Gamma approximation in \cite[eq ~(23)]{11154019Ghadi} exhibits minor accuracy discrepancies and yields slightly conservative estimates. As before, the largest deviations occur in the RIS scenario. This behavior can be explained by the fact that the EC depends on the entire distribution profile, rather than being dominated by the extreme tails, as is the case for OP or ABER. Consequently, moderate inaccuracies in the tail regions have a limited impact on the EC, making the Gamma approximation a reliable and computationally efficient alternative for EC evaluation in both RIS- and FRIS-assisted wireless systems.

Finally, Figs.~\ref{fig: EC 25} and~\ref{fig: EC 36} clearly demonstrate that the FRIS-assisted system consistently outperforms the conventional RIS-assisted counterpart. Moreover, increasing the number of active elements leads to further performance improvements, as reflected by higher EC values, in line with theoretical expectations.

\begin{figure}[t!]
\centering
\includegraphics[width=0.8\linewidth]{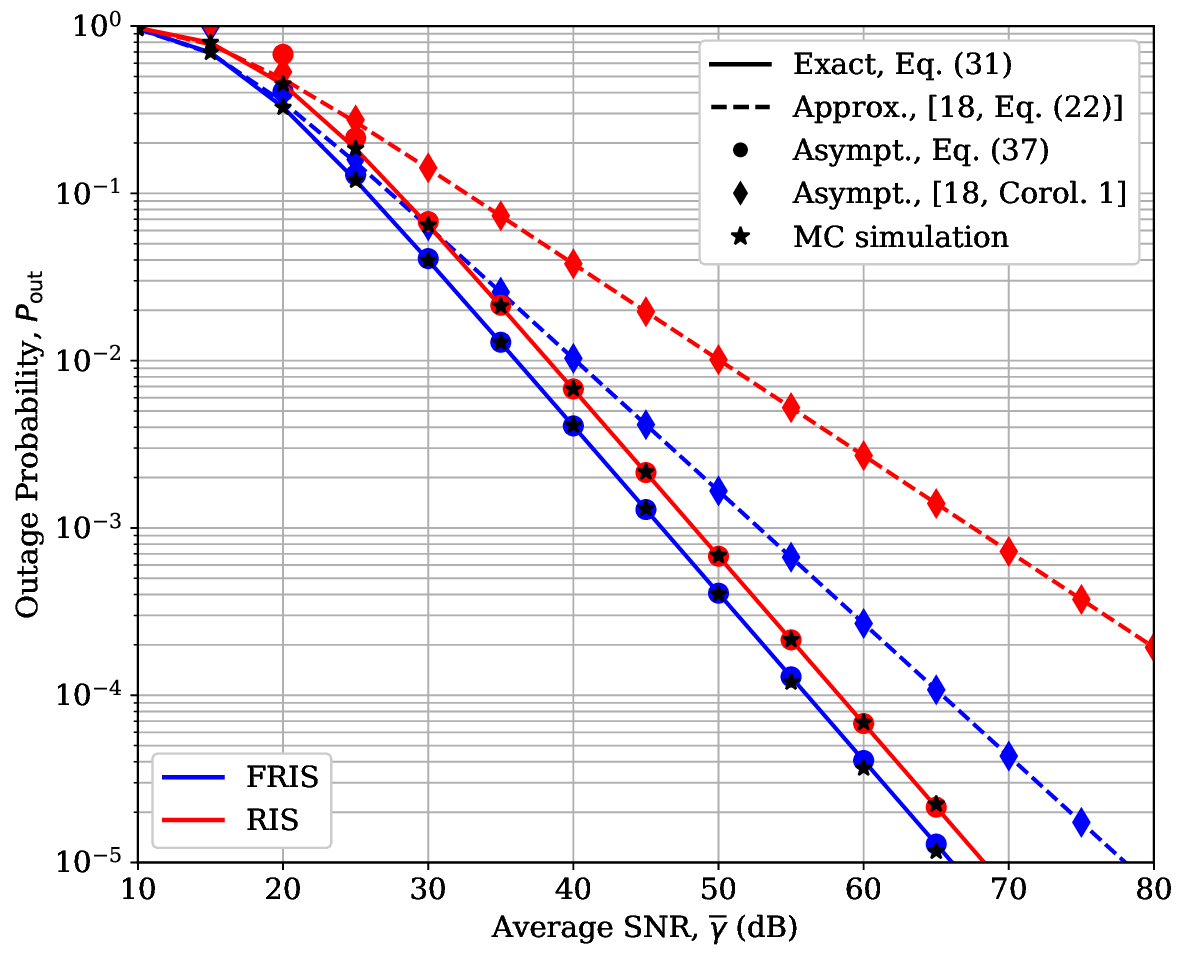}
\caption{OP versus the average SNR, $\bar{\gamma}$, for different FRIS and RIS configurations with $M_{\mathrm{on}} = 25$.}
\label{fig:asymptotic25}
\end{figure}

\begin{figure}[t!]
\centering
\includegraphics[width=0.8\linewidth]{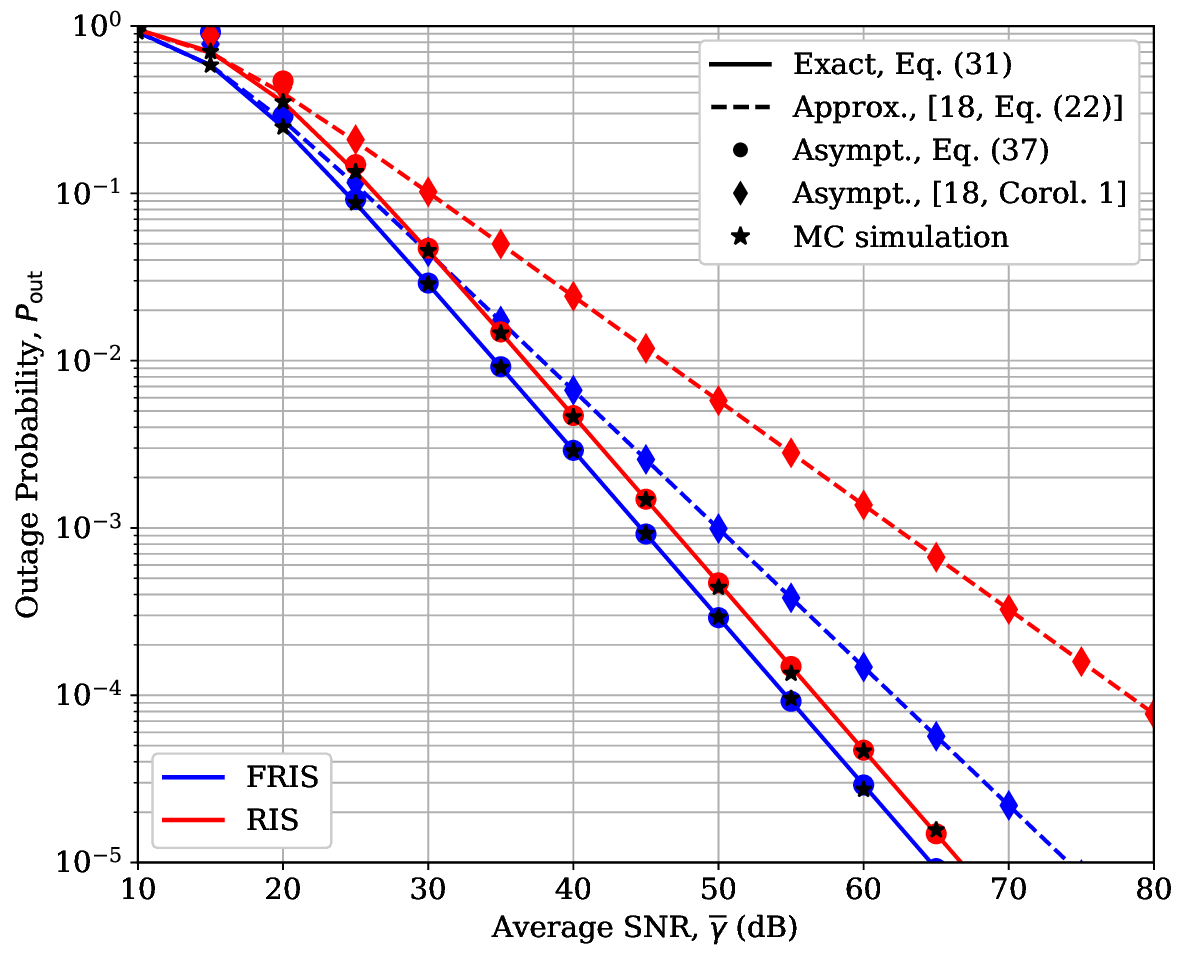}
\caption{OP versus the average SNR, $\bar{\gamma}$, for different FRIS and RIS configurations with $M_{\mathrm{on}} = 36$.}
\label{fig:asymptotic36}
\end{figure}

\begin{figure}[t!]
\centering
\includegraphics[width=0.8\linewidth]{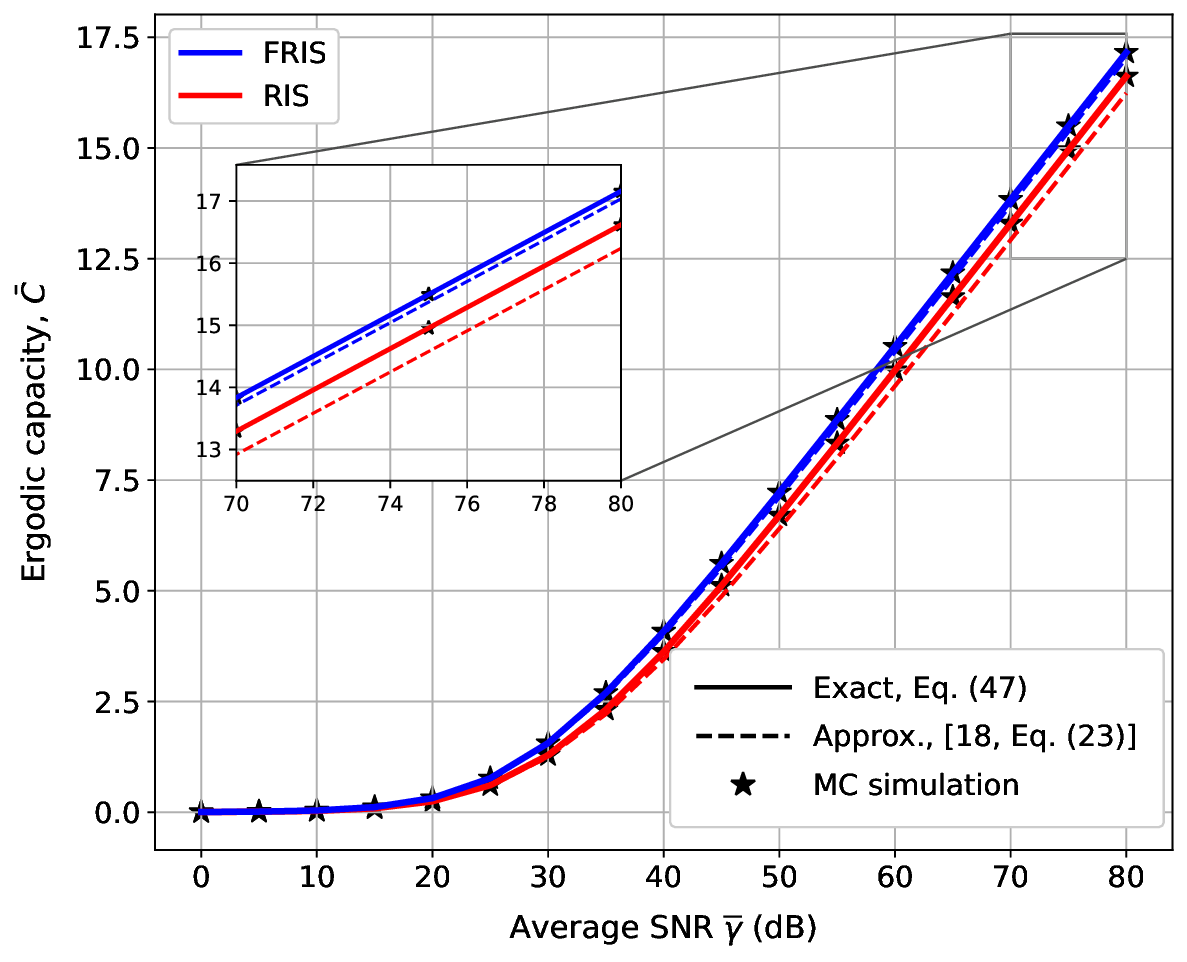}
\caption{EC versus the average SNR, $\bar{\gamma}$, for different FRIS and RIS configurations with $M_{\mathrm{on}} = 25$.}
\label{fig: EC 25}
\end{figure}

\begin{figure}[t!]
\centering
\includegraphics[width=0.8\linewidth]{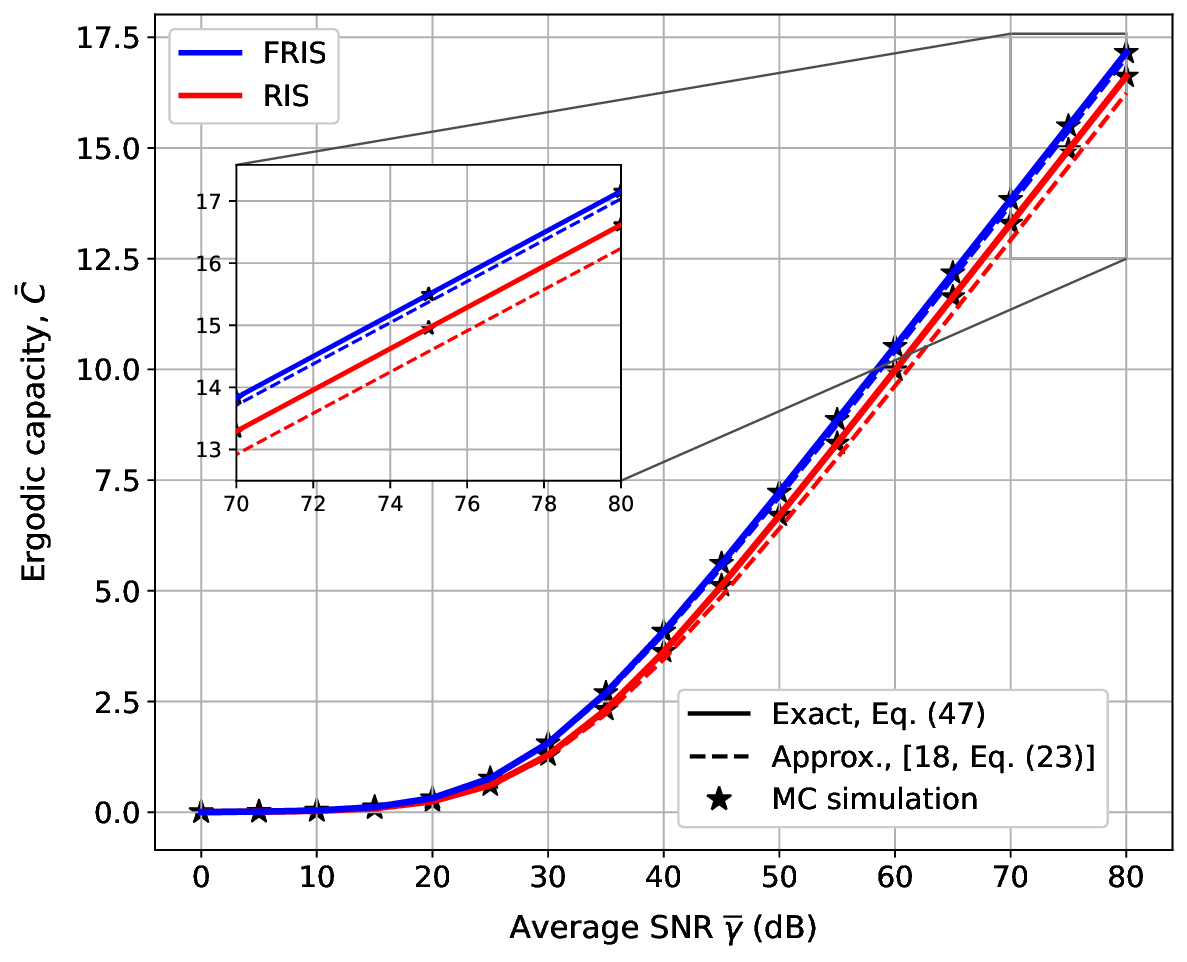}
\caption{EC versus the average SNR, $\bar{\gamma}$, for different FRIS and RIS configurations with $M_{\mathrm{on}} = 36$.}
\label{fig: EC 36}
\end{figure}

\section{Conclusion}
\label{sec: Conclusions}
This work developed an exact analytical framework for FRIS-assisted wireless systems by characterizing the cascaded channel gain under arbitrary spatial correlation. Exploiting the eigenstructure of the FRIS-induced correlation matrix, we showed that the channel gain admits a finite mixture of $K$-distributions. This provides, to the best of the authors’ knowledge, the first exact closed-form statistical characterization of FRIS-assisted systems and, as a special case, correlated RIS-assisted systems.
The framework unifies multiple operating regimes---from fully correlated to effectively and intrinsically uncorrelated---each with a clear physical interpretation. Using the exact statistics, we derived closed-form expressions for the OP and EC. The asymptotic OP analysis reveals a universal diversity order of one, independent of spatial correlation or fluidic reconfiguration, refining the diversity behavior reported in \cite[Corollary~1]{11154019Ghadi}.
Numerical results validate the analytical expressions, demonstrate improved accuracy over \cite{11154019Ghadi}, and show that FRIS architectures can significantly outperform conventional RIS by exploiting physical reconfiguration to suppress spatial correlation.

% This work presented a comprehensive analytical and numerical study of
% FRIS and compared their performance
% with \ac{RIS}. The proposed
% formulation derived the exact distribution of the cascaded channel gain
% $G_0 = |\mathbf{g}_\text{u}^\text{H}\mathbf{A}\mathbf{g}_\text{f}|^2$ using a finite linear combination of $K$-distributions and verified it through Monte Carlo simulations. The analytical and simulation results were shown to be in perfect agreement, confirming the accuracy of the proposed model.
% Overall, these results highlight the potential of FRIS-based architectures to enhance the robustness and efficiency of reconfigurable wireless links. We obtained closed-form expressions for the PDF and CDF of the effective channel using linear combination of $K$-distributions, offering clear insight into FRIS-assisted channel behavior. We also derived analytical formulas for the outage probability and ergodic capacity, along with high-SNR asymptotic characterizations. Numerical evaluations show that FRIS can markedly enhance both reliability and spectral efficiency compared with conventional RIS. Overall, these results highlight FRIS as a scalable and adaptable enabling technology for next-generation wireless networks

\begin{appendices}
\section{Proof of the Theorem}
\label{app: Theorem}

Let $\mathbf{C} \triangleq \mathbf{A} \mathbf{A}^\text{H} \in\mathbb{C}^{M_{\mathrm{on}}\times M_{\mathrm{on}}}$. Throughout the analysis, $\boldsymbol{\Phi}$ is assumed to be arbitrary but fixed.
From \eqref{eq:A-theta} and since $\mathbf{R}_S$ is Hermitian and positive semidefinite, it can be shown that $\mathbf{C}$ is also Hermitian and positive semidefinite. Consequently, all eigenvalues of $\mathbf{C}$ are nonnegative, i.e., $\lambda_i \ge 0$.

Define $Z \triangleq \mathbf{g}_\text{u}^\text{H}\mathbf{A}\mathbf{g}_\text{f}$ and $T \triangleq \mathbf{g}_\text{u}^\text{H}\mathbf{C}\mathbf{g}_\text{u}$.
% Conditioned on $\mathbf{g}_\text{u}$, it follows that $Z| \mathbf{g}_\text{u} \sim \mathcal{CN}(0,T | \mathbf{g}_\text{u})$, where $T | \mathbf{g}_\text{u}$ is a deterministic quantity.
Upon conditioning on $\mathbf{g}_\text{u}$, the quantity $T | \mathbf{g}_\text{u}$ becomes deterministic.
% On the other hand, since $\mathbf{g}_\text{f}$ is independent of $\mathbf{g}_\text{u}$,
Moreover, conditioning on $\mathbf{g}_\text{u}$, the mean and variance of $Z$ are given, respectively, by
\par\nobreak\vspace{-\abovedisplayskip}
\small
\begin{align}
\mathbb{E}\!\left[ Z | \mathbf{g}_\text{u} \right]
&= \mathbb{E}\!\left[\mathbf{g}_\text{u}^\text{H}\mathbf{A}\mathbf{g}_\text{f} | \mathbf{g}_\text{u} \right] = \mathbf{g}_\text{u}^\text{H}\mathbf{A}\,\mathbb{E}\!\left[\mathbf{g}_\text{f}\right]= \mathbf{g}_\text{u}^\text{H}\mathbf{A}\mathbf{0} = 0,
\end{align}
\normalsize
and
\par\nobreak\vspace{-\abovedisplayskip}
\small
\begin{align}
    \mathbb{E}\!\left[|Z|^2 | \mathbf{g}_\text{u} \right]
    &= \mathbb{E}\!\left[
        \mathbf{g}_\text{u}^\text{H}\mathbf{A}\mathbf{g}_\text{f}\,
        \mathbf{g}_\text{f}^\text{H}\mathbf{A}^\text{H}\mathbf{g}_\text{u}
        | \mathbf{g}_\text{u}
    \right] \nonumber\\
    &= \mathbf{g}_\text{u}^\text{H}\mathbf{A}\,
       \mathbb{E}\!\left[\mathbf{g}_\text{f}\mathbf{g}_\text{f}^\text{H}\right]\,
       \mathbf{A}^\text{H}\mathbf{g}_\text{u}| \mathbf{g}_\text{u} \nonumber\\
    &\overset{(a)}{=} \mathbf{g}_\text{u}^\text{H}\mathbf{A}\mathbf{A}^\text{H}\mathbf{g}_\text{u} | \mathbf{g}_\text{u} \overset{(b)}{=} T | \mathbf{g}_\text{u},
\end{align}
\normalsize
where in step~(a) we used $\mathbb{E}[\mathbf{g}_\text{f}\mathbf{g}_\text{f}^\text{H}] = \mathbf{I}_{M_{\mathrm{on}}}$, and in step~(b) we used the definition $\mathbf{C} = \mathbf{A}\mathbf{A}^\text{H}$ and $T| \mathbf{g}_\text{u} = \mathbf{g}_\text{u}^\text{H}\mathbf{C}\mathbf{g}_\text{u}| \mathbf{g}_\text{u}$.
Consequently, $Z$ is a circularly symmetric complex Gaussian random variable with zero mean and variance $T | \mathbf{g}_\text{u}$, i.e., $Z | \mathbf{g}_\text{u} \sim \mathcal{CN}\!\left(0,\,T | \mathbf{g}_\text{u}\right)$. Furthermore, since $G_0 = |Z|^2$, it follows that, conditioned on $T$ (or equivalently on $\mathbf{g}_\text{u}$), $G_0$ follows an exponential distribution with PDF given by
\par\nobreak\vspace{-\abovedisplayskip}
\small
\begin{align}
    f_{G_0}(g | T= t)
    = \frac{1}{t}\exp\!\left(-\frac{g}{t}\right).
    \label{eq:PDF-exp-G0}
\end{align}
\normalsize

Note from \eqref{eq:PDF-exp-G0} that obtaining the unconditional PDF of $G_0$ requires knowledge of the distribution of $T$. To this end, we employ a spectral decomposition approach, as detailed next.

% Define $r \triangleq \operatorname{rank}\big(\mathbf{C}(\boldsymbol{\theta})\big) = \sum_{i=1}^{q} m_i$.
Define $r \triangleq \sum_{i=1}^{q} m_i$. Because $\mathbf{C}$ is Hermitian and positive semidefinite, there exists a 
unitary matrix 
\(
\mathbf{U}\in\mathbb{C}^{M_{\mathrm{on}}\times M_{\mathrm{on}}}
\)
(i.e., $\mathbf{U}^\text{H}\mathbf{U}=\mathbf{U}\mathbf{U}^\text{H}=\mathbf{I}$) such that
\par\nobreak\vspace{-\abovedisplayskip}
\small
\begin{align}
    \mathbf{C}
= 
\mathbf{U}\,
\mathrm{diag}\!\Big(
\underbrace{\lambda_1,\ldots,\lambda_1}_{m_1},
\ldots,
\underbrace{\lambda_q,\ldots,\lambda_q}_{m_q},
\underbrace{0,\ldots,0}_{M_{\mathrm{on}}-r}
\Big)
\,\mathbf{U}^\text{H},
\end{align}
\normalsize
where $\{\lambda_i\}_{i=1}^{q}$ denote the $q$ distinct positive eigenvalues of $\mathbf{C}$ with corresponding multiplicities $\{m_i\}_{i=1}^{q}$.

Define the rotated channel vector $\mathbf{z} \triangleq \mathbf{U}^\text{H}\mathbf{g}_\text{u}$. Since $\mathbf{g}_\text{u}\sim\mathcal{CN}(\mathbf{0},\mathbf{I}_{M_{\mathrm{on}}})$
and the circularly symmetric complex Gaussian distribution is invariant
under unitary transformations, it follows that $
\mathbf{z}\sim\mathcal{CN}(\mathbf{0},\mathbf{I}_{M_{\mathrm{on}}})$.

Let $z_{i,\ell}$ denote the $\ell$-th component of the rotated vector 
$\mathbf{z}=\mathbf{U}^\text{H}\mathbf{g}_\text{u}$ corresponding to the eigenvalue 
group associated with $\lambda_i$, where $\ell=1,\ldots,m_i$.  
Because unitary transformations preserve the circularly symmetric complex 
Gaussian distribution, each $z_{i,\ell}$ is an independent 
$\mathcal{CN}(0,1)$ random variable.  
Consequently, $|z_{i,\ell}|^{2}\sim\mathrm{Exp}(1)$, independently for all 
$i$ and $\ell$.  

With this notation, the quadratic form 
$T=\mathbf{g}_\text{u}^\text{H}\mathbf{C}\mathbf{g}_\text{u}$ can be written in the 
diagonalized basis as
\par\nobreak\vspace{-\abovedisplayskip}
\small
\begin{align}
\label{eq:T-basis}
T = \sum_{i=1}^{q} \lambda_i \sum_{\ell=1}^{m_i} |z_{i,\ell}|^{2}.
\end{align}
\normalsize

Note from \eqref{eq:T-basis} that $T$ is a finite weighted sum of independent
unit-mean exponential random variables, where each weight $\lambda_i$
corresponds to an eigenvalue of $\mathbf{C}$ and appears according to its
multiplicity $m_i$.  
Therefore, $T$ follows a generalized chi-square distribution, whose Laplace transform admits the compact representation~\cite{Moschopoulos85}
\par\nobreak\vspace{-\abovedisplayskip}
\small
\begin{align}
    \label{eq:LT-T}
    \Phi_T(s) = \prod_{i=1}^{q}(1+\lambda_i s)^{-m_i}.
\end{align}
\normalsize

Applying a partial-fraction decomposition to \eqref{eq:LT-T}, it follows that
\par\nobreak\vspace{-\abovedisplayskip}
\small
\begin{equation}
\Phi_T(s)
= \sum_{i=1}^{q}\sum_{k=1}^{m_i}
\frac{c_{i,k}}{(1+\lambda_i s)^{k}},
\label{eq:PF}
\end{equation}
\normalsize
where the coefficients $c_{i,k}$ correspond to the residues associated with the 
repeated poles at $s=-1/\lambda_i$. 
These coefficients $c_{i,k}$ can be obtained explicitly by applying the standard
repeated-pole residue formula~\cite{Kreyszig10}, namely,
\par\nobreak\vspace{-\abovedisplayskip}
\small
\begin{equation}
c_{i,k}
=
\frac{1}{(m_i-k)!}\,
\lim_{s\to -1/\lambda_i}
\frac{d^{\,m_i-k}}{ds^{\,m_i-k}}
\!\left[
(1+\lambda_i s)^{m_i} \,\Phi_T(s)
\right].
\label{eq:residues}
\end{equation}
\normalsize

To obtain an explicit expression for the coefficients $c_{i,k}$ in
\eqref{eq:PF}, we consider the most general setting in which the
eigenvalues $\{\lambda_i\}_{i=1}^{q}$ of $\mathbf{C}$ may exhibit
arbitrary multiplicities, i.e., $m_i>1$ for one or more indices $i$.
In this case, the Laplace transform of $T$ in \eqref{eq:LT-T}
has, for each $i$, a pole of order $m_i$ located at $s=-1/\lambda_i$.

Multiplying \eqref{eq:LT-T} by $(1+\lambda_i s)^{m_i}$ cancels the
pole of order $m_i$, yielding
\par\nobreak\vspace{-\abovedisplayskip}
\small
\begin{equation}
\label{eq:Gi-def}
G_i(s)
\triangleq (1+\lambda_i s)^{m_i}\Phi_T(s)
= \prod_{\substack{j=1 \\ j\neq i}}^{q}(1+\lambda_j s)^{-m_j},
\end{equation}
\normalsize
which is analytic at $s=-1/\lambda_i$.

Now, we introduce the local variable:
\par\nobreak\vspace{-\abovedisplayskip}
\small
\begin{align}
    \label{eq: local variables}
    u \triangleq 1+\lambda_i s \, 
\Longleftrightarrow \, 
s = \frac{u-1}{\lambda_i},
\end{align}
\normalsize
which, for $j\neq i$, we have $1+\lambda_j s
= r_j u + (1-r_j)$, with $r_j \triangleq \frac{\lambda_j}{\lambda_i}$.
Accordingly,
\par\nobreak\vspace{-\abovedisplayskip}
\small
\begin{equation}
\label{eq:Hi-def}
H_i(u)
\triangleq
G_i(s(u))
= \prod_{j\neq i}(r_j u + 1 - r_j)^{-m_j},
\end{equation}
\normalsize
which is analytic around $u=0$ and admits a convergent Taylor expansion.

The global partial-fraction decomposition of $\Phi_T(s)$ is given in \eqref{eq:PF}, which includes contributions from all poles at
$s=-1/\lambda_1,\ldots,-1/\lambda_q$.
However, when computing the coefficient $c_{i,k}$, we focus on the
\emph{local} behavior of $\Phi_T(s)$ near the specific pole
$s=-1/\lambda_i$.  In this neighborhood, only the terms
\(
(1+\lambda_i s)^{-k}
\)
are singular, while all other factors 
\(
(1+\lambda_j s)^{-m_j},\; j\neq i,
\)
remain analytic and thus contribute only regular (nonsingular) terms.
Therefore, near $s=-1/\lambda_i$, the local expansion reduces to
\par\nobreak\vspace{-\abovedisplayskip}
\small
\begin{equation}
\label{eq:local-PF}
\Phi_T(s)
= \sum_{k=1}^{m_i}
\frac{c_{i,k}}{(1+\lambda_i s)^{k}}
\;+\; \text{analytic terms}.
\end{equation}
\normalsize
% and involves only a single summation over $k$.
% This is a standard consequence of the Laurent expansion of rationalfunctions around an isolated pole and explains why the coefficients $c_{i,k}$ are obtained using a one-dimensional index.

Now, we multiply \eqref{eq:local-PF} by $(1+\lambda_i s)^{m_i}=u^{m_i}$ to obtain
\par\nobreak\vspace{-\abovedisplayskip}
\small
\begin{align}
    \label{}
    u^{m_i}\Phi_T(s(u))
= \sum_{k=1}^{m_i} c_{i,k}\,u^{m_i-k}.
\end{align}
\normalsize
Since $u^{m_i}\Phi_T(s(u)) = H_i(u)$, it follows from \eqref{eq:Hi-def} that
\par\nobreak\vspace{-\abovedisplayskip}
\small
\begin{equation}
\label{eq:Hi-expansion}
H_i(u)
= \sum_{k=1}^{m_i} c_{i,k}\,u^{m_i-k}.
\end{equation}
\normalsize
Thus, the coefficient of $u^{m_i-k}$ in the Taylor expansion of $H_i(u)$
equals $(m_i-k)!\,c_{i,k}$.
Finally, extracting the corresponding Taylor coefficient yields
\eqref{eq: coeff cik}, which provides the exact residues associated with the repeated pole at $s=-1/\lambda_i$ and remains valid for arbitrary eigenvalue
multiplicities.

Having obtained the coefficients $c_{i,k}$, we can now invert the partial-fraction
representation in \eqref{eq:PF}. Using the inverse Laplace transform identity in
\cite[Table~II, p.~210]{bookSchiff99}, the PDF of $T$ follows immediately as
\par\nobreak\vspace{-\abovedisplayskip}
\small
\begin{equation}
f_T(t)
= \sum_{i=1}^{q}\sum_{k=1}^{m_i}
\frac{c_{i,k} \, t^{k-1}}{\lambda_i^{k}\Gamma(k)} \exp \left( - \frac{t}{\lambda_i}\right).
\label{eq:fT}
\end{equation}
\normalsize

Now, the PDF of $G_0$ is obtained by averaging the conditional density over $T$, namely,
\par\nobreak\vspace{-\abovedisplayskip}
\small
\begin{align}
    \label{eq:averaging-over-T}
    f_{G_0}(g)
    = \int_{0}^{\infty} f_{G_0}(g \mid T=t)\, f_T(t)\, \mathrm{d}t.
\end{align}
\normalsize

Substituting \eqref{eq:PDF-exp-G0} and \eqref{eq:fT} into 
\eqref{eq:averaging-over-T}, and interchanging the order of summation and
integration, we obtain
\par\nobreak\vspace{-\abovedisplayskip}
\small
\begin{align}
    \label{eq:PDF-G0-inner-integral}
    \nonumber f_{G_0}(g)
    = &\sum_{i=1}^{q}\sum_{k=1}^{m_i}
       \frac{c_{i,k}}{\lambda_i^{k}\Gamma(k)}\\
       & \times  \int_{0}^{\infty}
       \exp\!\left(-\frac{g}{t}\right)
       t^{k-2}
       \exp\!\left(-\frac{t}{\lambda_i}\right)
       \mathrm{d}t.
\end{align}
\normalsize
The interchange of summation and integration is justified because $f_{G_0}(g\mid T=t)$ and $f_T(t)$ are non–negative functions for all
$t,g\ge 0$, and the resulting integrand is absolutely integrable.

The integral in \eqref{eq:PDF-G0-inner-integral} can be evaluated in closed form by
invoking \cite[eq.~(3.471.9)]{Gradshteyn2007}, which yields the PDF of $G_0$ given in
\eqref{eq:PDF_K}.
The CDF of $G_0$, on the other hand, follows by integrating $f_{G_0}(g)$ from $0$ to $g$, i.e.,
$\int_{0}^{g} f_{G_0}(u)\,\mathrm{d}u$. With the aid of the identity in
\cite[eq.~(9.6.28)]{Abramowitz1972}, this integral can also be evaluated in
closed form, resulting in \eqref{eq:CDF_K}.  
This completes the proof.

% the expression
% \begin{equation}
% \label{eq:cik-final}
% c_{i,k}
% = \frac{1}{(m_i-k)!}\,
% \frac{d^{\,m_i-k}}{du^{\,m_i-k}}
% H_i(u)\bigg|_{u=0}.
% \end{equation}
% Expression \eqref{eq:cik-final} provides the exact residues associated with the
% repeated pole at $s=-1/\lambda_i$ and is valid for arbitrary eigenvalue
% multiplicities.

\section{Proof of Corollary~\ref{sec: corollary 1}}
\label{app: corollary 1}

When all eigenvalues of $\mathbf{C}$ are simple, i.e., $m_i = 1$ for every $i$,
the Laplace transform of $T$ in \eqref{eq:LT-T} reduces to
\par\nobreak\vspace{-\abovedisplayskip}
\small
\begin{align}
\label{eq:LT-T-simple}
\Phi_T(s)
= \prod_{i=1}^{q} (1+\lambda_i s)^{-1}.
\end{align}
\normalsize
In this case, the partial-fraction expansion \eqref{eq:PF} contains only
first–order terms, and therefore $c_{i,k}=0$ for every $k \ge 2$.

For the simple pole at $s=-1/\lambda_i$ in \eqref{eq:LT-T-simple}, the
corresponding coefficient in the partial-fraction expansion is
\par\nobreak\vspace{-\abovedisplayskip}
\small
\begin{equation}
c_{i,1}
= \lim_{s\to -1/\lambda_i} (1+\lambda_i s)\,\Phi_T(s).
\label{eq:ci1-limit}
\end{equation}
\normalsize

Factoring out the $i$-th term, \eqref{eq:ci1-limit} becomes
\par\nobreak\vspace{-\abovedisplayskip}
\small
\begin{align}
\label{eq:Factoring-out}
\Phi_T(s)
= (1+\lambda_i s)^{-1}
  \prod_{\substack{j=1 \\ j\neq i}}^{q} (1+\lambda_j s)^{-1}.
\end{align}
\normalsize
Multiplying both sides of \eqref{eq:Factoring-out} by $(1+\lambda_i s)$ gives
\par\nobreak\vspace{-\abovedisplayskip}
\small
\begin{align}
(1+\lambda_i s)\,\Phi_T(s)
= \prod_{\substack{j=1 \\ j\neq i}}^{q} (1+\lambda_j s)^{-1}.
\end{align}
\normalsize

Evaluating the limit at $s = -1/\lambda_i$ yields, for each $j\neq i$,
\par\nobreak\vspace{-\abovedisplayskip}
\small
\begin{align}
1+\lambda_j s
= 1 - \frac{\lambda_j}{\lambda_i}
= \frac{\lambda_i - \lambda_j}{\lambda_i},
\end{align}
and therefore
\begin{align}
(1+\lambda_j s)^{-1}
= \frac{\lambda_i}{\lambda_i - \lambda_j}.
\label{eq:lambda-final}
\end{align}
\normalsize

Substituting \eqref{eq:lambda-final} into \eqref{eq:ci1-limit} yields the
closed-form coefficient in \eqref{eq:ci1-simple}. Together with the fact that
$c_{i,k}=0$ for all $k \ge 2$ in the simple-eigenvalue case, this directly leads
to the final PDF and CDF expressions in \eqref{eq:PDFsimple-final} and
\eqref{eq:CDFsimple-final}, respectively.  
This concludes the proof.

\section{Proof of Corollary~\ref{sec: corollary 2}}
\label{app: corollary 2}

Assume that all nonzero eigenvalues of $\mathbf{C}$ are identical, such that
$\mathbf{C}=\lambda_1\,\mathbf{I}_{r}$ on its $r$-dimensional support, where
$1 \le r \le M_{\mathrm{on}}$. Under this condition, the eigenvalue structure
reduces to $q=1$, with a single eigenvalue $\lambda_1$ of multiplicity
$m_1=r$.

From the eigen-decomposition described in Appendix~\ref{app: Theorem}, the first
$r$ components of the rotated vector $\mathbf{z}$ satisfy
$z_i \sim \mathcal{CN}(0,1)$, while the remaining $M_{\mathrm{on}}-r$ components
lie in the null space of $\mathbf{C}$ and therefore do not contribute to the
quadratic form $T$. Consequently,
\par\nobreak\vspace{-\abovedisplayskip}
\small
\begin{align}
T
= \mathbf{g}_\text{u}^\text{H}\mathbf{C}\mathbf{g}_\text{u}
= \lambda_1\sum_{i=1}^{r}|z_i|^{2}
\sim \lambda_1\,\mathcal{G}(r,1),
\end{align}
\normalsize
which implies that $T$ follows a Gamma distribution with shape parameter $r$
and scale parameter $\lambda_1$. Accordingly, the PDF of $T$ is given by
\par\nobreak\vspace{-\abovedisplayskip}
\small
\begin{equation}
f_T(t)
= \frac{t^{r-1}}{\lambda_1^{r}\Gamma(r)}\,
  \exp\!\left(-\frac{t}{\lambda_1}\right),
\qquad t \ge 0.
\label{eq:fT-equal}
\end{equation}
\normalsize
Conditioned on $T$, the cascaded channel gain satisfies
$G_0 | T \sim \mathrm{Exp}(1/T)$, as stated in \eqref{eq:PDF-exp-G0}. The unconditional PDF of $G_0$ therefore follows from the averaging integral in
\eqref{eq:averaging-over-T}. Substituting \eqref{eq:fT-equal} into
\eqref{eq:averaging-over-T} and applying \cite[eq.~(3.471.9)]{Gradshteyn2007}
yields the closed-form expression in \eqref{eq:PDF-cor2-final}.

Finally, integrating \eqref{eq:PDF-cor2-final} from $0$ to $g$ and invoking
\cite[eq.~(9.6.28)]{Abramowitz1972} gives the corresponding CDF expression in
\eqref{eq:CDF-cor2-final}. This completes the proof.

\section{Derivation of Ergodic Capacity}
\label{app: ergodic capacity}

\subsection{General-Multiplicity Case}

Plugging \eqref{eq:PDF_G0} into \eqref{eq:C-def} and interchanging the order of
integration and summation, the ergodic capacity can be written as
\par\nobreak\vspace{-\abovedisplayskip}
\small
\begin{align}
\bar C
&=
\sum_{i=1}^{q}\sum_{k=1}^{m_i}
c_{i,k}\,
\bar C_K(k,\lambda_i;\bar{\gamma} L_\text{f} L_\text{u}),
\label{eq:C_mix}
\end{align}
\normalsize
where
\par\nobreak\vspace{-\abovedisplayskip}
\small
\begin{align}
\label{eq: C bar def}
\nonumber \bar C_K &(k,\lambda_i;\bar{\gamma} L_\text{f} L_\text{u})\\
& \triangleq
\int_{0}^{\infty}
\log_2\!\left(1+\bar{\gamma} L_\text{f} L_\text{u}\, g\right)
f_K(k,\lambda_i;g)\,\text{d}g.
\end{align}
\normalsize

Replacing \eqref{eq:PDF_K} into \eqref{eq: C bar def} and using
$\log_2(x)=\ln(x)/\ln(2)$ yields
\par\nobreak\vspace{-\abovedisplayskip}
\small
\begin{align}
\nonumber
\bar C_K &(k,\lambda_i;\bar{\gamma} L_\text{f} L_\text{u})
=
\frac{2\,\lambda_i^{-\frac{k+1}{2}}}{\Gamma(k)\,\ln (2)} \\
& \times
\int_0^\infty
\ln(1+\bar{\gamma} L_\text{f} L_\text{u} g)\,
g^{\frac{k-1}{2}}\,
K_{k-1}\!\left(2\sqrt{\frac{g}{\lambda_i}}\right)\,
\text{d}g.
\label{eq:Ck_start}
\end{align}
\normalsize

Next, performing the change of variables $g=\lambda_i x$ with
$\text{d}g=\lambda_i\,\text{d}x$, we obtain
\par\nobreak\vspace{-\abovedisplayskip}
\small
\begin{align}
\bar C_K &(k,\lambda_i;\bar{\gamma} L_\text{f} L_\text{u})
=
\frac{2}{\Gamma(k)\,\ln (2)}\\
& \times
\int_0^\infty
\ln(1+\bar{\gamma} L_\text{f} L_\text{u} \lambda_i x)\,
x^{\frac{k-1}{2}}\,
K_{k-1}(2\sqrt{x})\,\text{d}x.
\label{eq:Ck_norm}
\end{align}
\normalsize

To proceed, we employ the Mellin–Barnes representation of the logarithm function
(valid for $x>0$)~\cite{GradshteynRyzhik2000}
\par\nobreak\vspace{-\abovedisplayskip}
\small
\begin{equation}
\ln(1+x)
=
\frac{1}{2\pi \ii}
\int_{\mathcal{L}}
\frac{\Gamma(s)\Gamma(1-s)}{s}\,x^s\,\text{d}s,
\label{eq:MB_log}
\end{equation}
\normalsize
where $\mathcal{L}$ denotes a Bromwich contour separating the poles of
$\Gamma(s)$ and $\Gamma(1-s)$.

Substituting \eqref{eq:MB_log} into \eqref{eq:Ck_norm} and interchanging the
order of integration yields
\par\nobreak\vspace{-\abovedisplayskip}
\small
\begin{align}
\nonumber
\bar C_K &(k,\lambda_i;\bar{\gamma} L_\text{f} L_\text{u})
=
\frac{1}{\Gamma(k)\,\ln (2)} \\
& \times
\frac{1}{\pi \ii}
\int_{\mathcal{L}}
\frac{\Gamma(s)\Gamma(1-s)}{s}
\left(\bar{\gamma} L_\text{f} L_\text{u} \lambda_i\right)^s
I(s)\,\text{d}s,
\label{eq: contour C bar}
\end{align}
\normalsize
where $I(s)
\triangleq
\int_0^\infty
x^{s+\frac{k-1}{2}}\,
K_{k-1}(2\sqrt{x})\,\text{d}x$.
% \begin{equation}
% I(s)
% \triangleq
% \int_0^\infty
% x^{s+\frac{k-1}{2}}\,
% K_{k-1}(2\sqrt{x})\,\text{d}x.
% \label{eq:Is_def}
% \end{equation}

Applying the change of variables $y=2\sqrt{x}$ and using
\cite[eq.~(6.561.16)]{Gradshteyn2007}, the integral $I(s)$ admits the closed form
\par\nobreak\vspace{-\abovedisplayskip}
\small
\begin{equation}
I(s)
=
\frac{1}{2}\,
\Gamma(s+k)\Gamma(s+1).
\label{eq:Is_closed}
\end{equation}
\normalsize

Substituting \eqref{eq:Is_closed} into \eqref{eq: contour C bar} and using the
identity $\Gamma(s+1)=s\Gamma(s)$, the ergodic capacity term simplifies to
\par\nobreak\vspace{-\abovedisplayskip}
\small
\begin{equation}
\begin{aligned}
\bar C_K &(k,\lambda_i;\bar{\gamma} L_\text{f} L_\text{u})
=
\frac{1}{\Gamma(k)\,\ln (2)} \\
& \times
\frac{1}{2\pi \ii}
\int_{\mathcal{L}}
\Gamma(s)^2\,
\Gamma(1-s)\,
\Gamma(s+k)\,
\left(\bar{\gamma} L_\text{f} L_\text{u} \lambda_i\right)^s
\text{d}s.
\end{aligned}
\label{eq:Ck_MB_final}
\end{equation}
\normalsize

Recognizing the complex integral representation in \eqref{eq:Ck_MB_final} as a Meijer-$G$ function
\cite[eq.~(07.34.02.0001.01)]{Mathematica}, it follows that
\par\nobreak\vspace{-\abovedisplayskip}
\small
\begin{align}
\nonumber \bar C_K(k,\lambda_i;\bar{\gamma} L_\text{f} L_\text{u})
& =
\frac{1}{\ln (2) \Gamma(k)}\,
\\
& \times
G_{4,2}^{1,4}\!\left(
\bar{\gamma} L_\text{f} L_\text{u} \lambda_i\;\Bigg|\;
\begin{matrix}
1-k,\;0,\;1,\;1\\
1,\;0
\end{matrix}
\right).
\label{eq:Ck_meijerg}
\end{align}
\normalsize

Finally, substituting \eqref{eq:Ck_meijerg} into \eqref{eq:C_mix} yields the exact closed-form expression for the EC, as given in \eqref{eq:C_total_meijerg}.

\subsection{Simple-Eigenvalue Case}
An exact closed-form expression for the EC in the simple-eigenvalue case is readily obtained from \eqref{eq:C_total_meijerg} by setting $m_i=1$ for all $i$, which yields \eqref{eq:C_total_meijerg_simple}.

\subsection{Equal-Eigenvalue Case}
In this case, the composite matrix $\mathbf{C}$ admits a single distinct eigenvalue, i.e., $q=1$ with eigenvalue $\lambda_1$ of multiplicity $m_1$. Accordingly, the ergodic capacity expression in \eqref{eq:C_total_meijerg_simple} simplifies and reduces to \eqref{eq:EC-final-equal}.

\subsection{Uncorrelated Case}
When $\mathbf{C}=\mathbf{I}_{M_{\mathrm{on}}}$, the effective gain $G_0$ follows a
single $K$-distribution with shape parameter $k=m_1=M_{\mathrm{on}}$ and unit
scale parameter $\lambda_1=1$. Accordingly, \eqref{eq:EC-final-equal} directly applies, and
the EC simplifies to
\eqref{eq:EC-final-uncorrelated}. This concludes the derivation.

\end{appendices}

\bibliographystyle{IEEEtran}
\bibliography{references} % no .bib extension

\end{document}